\documentclass[11pt]{IEEEtran2e} %
\pdfoutput=1

\makeatletter
\def\ps@headings{%
\def\@oddhead{\mbox{}\scriptsize\rightmark \hfil \thepage}%
\def\@evenhead{\scriptsize\thepage \hfil \leftmark\mbox{}}%
\def\@oddfoot{}%
\def\@evenfoot{}}
\makeatother
\usepackage{amsmath, amssymb}
\usepackage{graphicx}
\usepackage{dsfont}
\usepackage[usenames,dvipsnames]{xcolor}
\usepackage{ifsym, wasysym}
\usepackage{verbatim}
\usepackage[tight,footnotesize]{subfigure}
\usepackage{cite}
\usepackage{url}
\newtheorem{theorem}{Theorem}
\newtheorem{corollary}{Corollary}
\newtheorem{lemma}{Lemma}

\def\deq{{\,{\buildrel \bigtriangleup \over =}\,}}

\def\conv{\otimes}
\def\deconv{\oslash}
\def\eps{\varepsilon}
\def\P{{Pr}}  
\def\minplus{{$(\min, +)\,$}} 
\def\S{{\cal S}}
\def\A{{\cal A}}
\def\D{{\cal D}}
\def\B{{\cal B}}
\def\W{{\mathcal W}}
\def\M{{\mathcal M}}
\def\F{{\mathcal F}}
\def\X{{\mathcal X}}
\def\Y{{\mathcal Y}}
\def\Z{{\mathcal Z}}
\def\Z{{\mathcal Z}}
\def\Msum{{\mathsf{M}}}
\def\mx{{$(\min, \times)$}}
\def\sr{{$(\sigma(s), \rho(s))$}}
\graphicspath{{./figs/}}

\begin{document}

\title{A Network Calculus  Approach for the Analysis of  Multi-Hop    Fading  Channels} 
\author{Hussein~Al-Zubaidy*, 
         J\"{o}rg Liebeherr*, Almut Burchard**  \\ 
        * Department of ECE,   University of Toronto,  Canada. \\
        ** Department of Mathematics,   University of Toronto, Canada. \\
	 	E-mail:  \{hzubaidy, jorg\}@comm.utoronto.ca; almut@math.toronto.edu.
        }%

\maketitle

\begin{abstract} 
A fundamental problem for the delay and backlog analysis
across multi-hop paths in wireless networks 
is how to account for the random properties of the 
wireless channel. Since the usual statistical models 
for radio signals in a propagation environment 
do not lend themselves easily to a description of 
the available service rate,
the performance analysis of wireless networks 
has resorted to higher-layer abstractions, 
e.g., using Markov chain models. 
In this work, we propose a network calculus 
that can incorporate common statistical models of 
fading channels and obtain statistical bounds on 
delay and backlog across multiple nodes.  
We conduct the analysis in a transfer domain, which we 
refer to as the {\it SNR domain}, where the service process
at a link is characterized by the instantaneous signal-to-noise ratio
at the receiver.  We discover that, in the transfer domain, 
the network model is governed by a dioid algebra, which we refer to as 
\mx~algebra. Using this algebra we derive the desired delay and backlog bounds. 
An application of the analysis is demonstrated 
for a simple multi-hop network with Rayleigh fading channels. 
\end{abstract}

%
%
%

\section{Introduction}
The network-layer  performance analysis  
 seeks to provide estimates on the delays
experienced by traffic traversing the elements of a network,
as well as the corresponding buffer requirements.
For wireless networks, a question of interest is how the stochastic 
properties of  wireless channels impact delay and
backlog performance.
Wireless channels are characterized by  
rapid variation of  channel quality 
caused by the mobility and location of  communicating devices. This is due to {\it fading}, which is the deviation in the  attenuation  experienced by the transmitted signal  when traversing a wireless  channel. The term {\it fading channel} is used to refer to  a channel that experiences such effects. 
In this paper we explore the network-layer performance of 
a multi-hop  network where each link is 
represented by a fading channel.

We model the wireless network by tandem queues with variable capacity 
servers, where each server expresses the random capacity of a 
fading channel. We ignore the impact of coding by assuming that 
transmission rates  over the fading channels are equal to their information-theoretic 
capacity limit, $C$, which is expressed as a function of the instantaneous 
signal-to-noise ratio (SNR) at the receiver,  $\gamma$, 
by  
$C(\gamma) = W \log (1+ \gamma)$,  
where $W$ is the channel bandwidth (in {\it Hz}).  
Numerous models are available to describe the gain of fading channels depending on the type of fading (slow or fast), 
and the environment (e.g., urban or rural). 
The instantaneous, information-theoretic  channel capacity of a fading channel can be represented as the logarithm of  $\gamma$ by (see Chp.~14.2 in  \cite{Proakis2007})  
\begin{equation} \label{eq:capacity}
C(\gamma) = c \log \big( g(\gamma) \big) \, ,
\end{equation} 
where $c$ is a constant and the function $g(\gamma)$ is  used  to characterize the fading channel. 
We are interested in finding bounds on the end-to-end delay and on buffer 
requirements for a cascade of fading channels, with store-and-forward 
processing at each channel. 
 
The analysis in this paper takes a system-theoretic 
stochastic network calculus approach \cite{Jiang-Book}, which 
describes the network properties using a  $(\min, +)$ dioid algebra.
Arrival and departure processes at a network element are 
described by bivariate stochastic processes $A(\tau,t)$ and 
$D(\tau,t)$, respectively, denoting the 
cumulative arrivals (departures) in the time interval $[\tau,t)$. A  
 network element is characterized by the service process $S(\tau,t)$, 
denoting the 
available service in $[\tau,t)$. The input-output relationship at the network 
element is governed by 
\begin{align} \label{eq:conv}
D(0,t) \geq A \ast S (0,t) \, , 
\end{align} 
where the \minplus convolution operator `$\ast$'  is defined as   
 $f \ast g (\tau, t) = \inf_{\tau \leq u \leq t} 
\{ f (\tau,u) + g(u, t) \}$.
If network traffic passes through a tandem of $N$ network elements with service 
processes $S_1,S_2, \ldots, S_N$, the service of the network as a whole can be 
expressed by the convolution $S_1 \ast S_2\ast \ldots \ast S_N$.

The stochastic properties of fading channels present a 
formidable challenge for a network-layer analysis  
since the service processes corresponding
to the channel capacity of common fading channel models 
such as  Rician, Rayleigh, or Nakagami-$m$, 
require to take a logarithm of their distributions.
As discussed in the next section, researchers  
frequently turn to higher-layer abstractions to model fading channels. 
Widely used abstraction are  the two-state Gilbert-Elliott model and its 
extensions to finite-state Markov channels (FSMC) \cite{Sadeghi}.
FSMC models simplify the analysis to a degree that the network model becomes 
tractable, at least at a single node. Extensions to multi-hop settings 
encounter a rapidly growing state space. As of today, a general multihop analysis 
that is applicable to  models of  fading channels, such as Rician, Rayleigh, or Nakagami-$m$, remains open.

In this paper, we pursue a  novel approach 
to the analysis of multi-hop wireless networks. We develop a 
calculus for wireless networks that can be applied to  
fading channel models from the wireless communication 
literature to provide network-layer performance bounds.
We view the  network-layer model with arrival, departure 
and service processes 
as residing in a {\it  bit domain}, where traffic and 
service is measured in 
bits. 
We view the fading channel models used in wireless communications 
as residing in an alternate domain, which we call the {\it  SNR domain}, 
where channel properties are expressed in terms of the distribution of the 
signal-to-noise ratio at the receiver. 
We then derive a method to compute performance bounds from 
these traffic and service characterizations.

A key observation in our work is that service elements in the 
SNR domain obey the laws of a dioid algebra.  We devise a 
suitable dioid, referred to as {\it  \mx~algebra}, where the minimum 
takes the role of the standard addition, and the second operation
is the usual multiplication, and use it for analysis 
in the SNR domain. In this domain  multi-hop descriptions of fading channels  become  tractable. 
In particular, we find that a cascade of fading channels can be expressed 
in terms of a convolution in the new algebra of the constituting channels. 
The key to our analysis is that we derive   
performance bounds entirely in the SNR domain. Observing that the 
bit and SNR domains are linked by the exponential function,
we transfer arrival and departure 
processes from the bit to the SNR domain. 
Then, we derive backlog and delay bounds in the transfer domain using 
the \mx~algebra. The results are  mapped back to the original bit domain 
to finally give us the desired performance bounds.
Our derivations in the SNR domain require the computation 
of products and quotients of random variables. Here, we take 
advantage of the  \textit{Mellin transform} 
to facilitate  otherwise cumbersome calculations.  Then, the computational 
problem is reduced  to finding the Mellin transform for  service and 
traffic processes.  

The main contribution of this paper is the development of a 
framework for studying the impact of channel gain models 
on the network-layer performance of wireless networks. 
For the purposes of this paper, the SNR domain is used solely as a 
transfer domain that enables us to solve an otherwise intractable 
 mathematical problem. On the other hand, the  ability to map     
quantities that appear in network-layer models and concepts 
found in a physical-layer analysis may prove useful in 
a broader context, e.g., for studying 
cross-layer performance issues in wireless communications. 
Moreover, the \mx~algebra and the Mellin transform form a tool set 
that can be applied more generally in wireless communications 
for studying the channel gain of cascades of fading channels. 
As the first paper on the \mx~network calculus algebra, our 
paper only considers simple network scenarios and makes 
numerous convenient assumptions (which are made explicit 
in Sec.~\ref{sec:model}). 
There is room for significant future work on 
extensions of the model and a relaxation of the presented 
assumptions.   

The remainder of the paper is organized as follows. 
In Sec.~\ref{sec:related} we discuss related work. 
We describe the system model in Sec.~\ref{sec:model}, where we also  
motivate the use of the SNR domain.  In Sec.~\ref{sec:calculus} we present the \mx~algebra and derive performance bounds. 
In Sec.~\ref{sec:rayleigh} we apply the
analysis to a cascade of Rayleigh channels, and present 
numerical examples. In Sec.~\ref{sec:rayleigh}, we investigate 
a network with cross traffic at each node. 
We discuss brief conclusions in Sec.~\ref{sec:conclusions}.


\section{Related Work}
\label{sec:related}

Analytical approaches for network-layer performance analysis of wireless networks 
include queueing theory, effective bandwidth and, more recently, 
network calculus. Since the service processes corresponding
to the channel capacity of common fading channel models 
such as  Rician, Rayleigh, or Nakagami-$m$, 
require to take a logarithm of their distributions, researchers    
often turn to higher-layer abstractions to model fading channels, 
which lend themselves more easily to an analysis. 
A widely used abstraction is  the two-state channel model developed 
by Gilbert \cite{Gilbert} and Elliott \cite{Elliott}, and 
subsequent  extensions to  a finite-state Markov 
channel (FSMC) \cite{WangMoayeri}. Markov channel models are 
well suited to express the time correlation of fading channel samples. We refer to \cite{Sadeghi} for a survey of the development and applications of FSMC models.
Zorzi et al. \cite{Zorzi-ICPUC} evaluated the accuracy of first-order Markov channel models of fading channels,  where the next channel sample depends only on the current state of the Markov process, and higher order processes that can capture memory  extending further back in the process history.  The authors found that a first-order Markov model is a good approximation of the fading channel, and that using higher order Markov processes does not significantly  
improve the accuracy of the model.  

Queue-based channel (QBC) \cite{Zhong_Alajaji_Takahara} is an alternative model for fading channels, which 
models a binary additive noise channel with memory based on a finite queue. Here,  a queue with size $M$ contains the  last $M$ noise symbols, 
and the noise process is an $M^{th}$-order Markov chain.  The model was found to  provide a better approximation to the Rayleigh and Rician slow fading channels compared to the Gilbert-Elliot model  \cite{Ngatched}.  An extension of the QBC model, called Weighted QBC \cite{Ngatched} permit queue cells (i.e., channel samples) to contribute with different weights to the noise process.

Queueing theoretic studies of fading channels
generally apply  approximations to reduce the complexity of multi-hop models.  Le, Nguyen, and Hossain \cite{Le_Nguyen_HossainWCNC07} apply a decomposition approximation to analyze the loss probability and average delay of a multi-hop wireless network with slotted transmissions for a batch Bernoulli arrival process, and with independent cross traffic at each node. The wireless link is assumed to employ adaptive modulation and coding with multiple modes, where each mode corresponds to a given link rate, which is chosen based on the SNR of the channel.  The channel state of a link is assumed to be stationary, and channel states in consecutive time slots 
are independent.  
Another decomposition approximation is presented by Le and Hossain \cite{Le_Ekram08}, who consider a multi-hop  tandem network with a batch arrival process and multi-rate transmissions, to develop a routing scheme that can meet given 
delay and loss requirements. 
The analysis obtains end-to-end loss rates and delays with a decomposition analysis, and feeds the results as metrics to the routing algorithm. 
Bisnik and Abouzeid \cite{Bisnik} model a multi-hop wireless network as 
an  open network of G/G/1 queueing systems. Using diffusion approximation,   
they obtain closed-form expressions for average end-to-end delays.  
Ishizaki and Hwang \cite{ishizaki2007} studied
the impact of multiuser diversity assisted packet scheduling on the packet delay performance in a 
wireless network with Nakagami-m fading channels. The network is modeled by a multi-queue system, where each channel is described by an FSMC. Under assumptions of stationarity, homogeneity, and independence of the channel processes, they approximate the tail distribution of 
packet delays, and compare it to that of a round-robin scheduler. The results indicate that 
the delay performance of multiuser diversity assisted scheduling algorithms is not necessarily superior to that of
round-robin scheduling. 
Since the application of queuing theoretical methods to study 
cascade of fading channels requires  
many  assumptions on  arrival and service distributions, and 
simplifications of the model, the use 
of classical queueing theoretic methods for the 
performance analysis  of multi-hop wireless networks  has been put into question \cite{Chen_Yang_Darwazeh}. 

An effective bandwidth \cite{Kelly} 
analysis seeks to develop (asymptotic) bounds on 
performance metrics, e.g., an exponential decay
of the backlog.   
Wu and Negi \cite{Wu} 
have adapted an effective bandwidth analysis to the analysis of 
fading channels. They  introduce the concept of effective capacity,   
which characterizes a wireless channel by a log-moment generating function (log-MGF) of the channel capacity. 
They  
obtain an asymptotic approximation of the delay bound violation probability of a  Rayleigh fading channel. 
Due to the difficulty of  computing the moment generating function (MGF)  of the Rayleigh distribution, they assume non-correlated distributions  with low SNR 
and estimate channel parameters from  measurements. 
The work has been extended to correlated Rayleigh  and correlated Nakagami-$m$  channels, and to cascades of fading channels 
\cite{Wang_CorrRayleigh, Wang_CorrNakagami, Wu-MNA06}.
A closely related concept is the effective channel capacity presented by Li et al. \cite{Chengzhi-Che-Li}, which describes the available channel capacity  by a first order Markov chain and computes the  log-MGF of the underlying Markov process. Taking advantage of 
methods developed in \cite{LiBuLi07}, they   
compute statistical delay bounds for Nakagami-m fading channel.  Hassan, Krunz, and Matta \cite{krunz2004} use an effective bandwidth analysis  
to study  delay
and  loss performance at a single wireless link, which is modeled by an FSMC.  For fluid On-Off traffic and FIFO
buffering, they obtain a closed form expression for 
the effective bandwidth required to guarantee bounds on delay and packet loss. 

There is a collection of recent works that 
apply stochastic network calculus  methods \cite{Jiang-Book}
to wireless networks with fading channels. 
The stochastic network calculus is closely related to the 
effective bandwidth theory, in that 
it seeks to develop bounds on performance metrics under   
assumptions also found in the effective bandwidth literature. 
Different from effective bandwidth literature, 
stochastic network calculus methods seek to develop non-asymptotic bounds. 
An attractive element of a network calculus analysis is 
that sometimes it is possible to 
extend a single node analysis to a tandem of nodes, using 
the \minplus convolution operation seen in the introduction. 

Fidler \cite{Fidler-Fading} presents a network  calculus methodology 
for a two-state  FSMC model of a  single-hop  fading channel. 
He applies the MGF network  calculus, which was 
suggested in  the problem sets of Chapter~7 in \cite{CSChang}, and 
which has been developed in \cite{FidlerMGF}. 
The MGF network  calculus  takes its name from 
the  extensive use of  moment generating functions in 
the derivation of performance bounds. 
Mahmood, Rizk, and Jiang \cite{Mahmood_Rizk_Jiang} apply the MGF network calculus to  MIMO 
channels   and derive delay bounds 
for periodic traffic sources. 
Zheng et. al. \cite{ZhengIET2011} also use an MGF network calculus
to study the performance of two-hop relay networks. 
A similar methodology is applied by Mahmod, Vehkaper\"{a}, and Jiang  \cite{MahmoodArxiv} to compute the 
throughput of a multi-user DS-CDMA system with delay constraints. 
In the works above that apply the MGF network  calculus,  
models for a cascade of multiple fading channels become  
complex, so that multi-node results for networks with 
more than two nodes have not been obtained. 
The \mx~network calculus developed in this paper 
uses similar descriptions and assumptions for traffic and service 
as the MGF network calculus. By performing  computations in a transfer domain, 
where  fading channel models take a simpler form, we are able to 
compute multi-node service descriptions for an arbitrarily large number of nodes.  
 
The MGF network calculus  assumes that arrivals and service at each node are independent. 
These  assumptions can be  relaxed  using statistical envelope descriptions for traffic (effective envelopes) and service (statistical 
service curve) \cite{Burchard_ToIT06,Jiang-Book}. 
Jiang and Emstad \cite{Jiang:2005} have applied an approach with envelopes to a fading channel where the wireless channel is 
characterized by two stochastic processes: an ideal service process and 
an impairment process, where the impairment process 
captures effects due to fading, noise, and cross traffic. 
Verticale and Giacomazzi \cite{Verticale:2009}  
obtain a closed form expression for the variance of a service curve, which 
describes the available service by a  Markov chain. This is used for the analysis of an FSMC model of a Rayleigh fading channel.  For computing the bounds for 
Markovian arrivals, they apply the bounded-variance network calculus  introduced in \cite{Giacomazzi-Saddemi}, which is an extension of the  central limit theorem methods  by Choe and Shroff \cite{Choe-Shroff-MVA} to multi-hop paths. 
Verticale  \cite{Verticale-Q2SWinet} has applied the same methodology 
to constant bit rate traffic. 
Ciucu, Pan, and Hohlfeld \cite{CiucuAlerton10} and Ciucu \cite{SigmetricsCiucu11} 
present non-asymptotic (i.e., finite number of hops), closed-form expressions for the delay and throughput distributions for multi-hop wireless networks. As many of the works discussed above, 
the fading channel is modeled by an abstraction  
that uses a link layer model of the transmission channel. Here,  the channel is assumed to be governed by a slotted-ALOHA system in half-duplex mode. The model of this channel is a  two-state On-Off server, 
where a node can transmit (i.e., is in the On state) only when 
other all nodes in the interference range are not transmitting. 

There is also a literature on physical-layer performance metrics  of fading channels in  multi-hop wireless  networks. Hasna and Alouini \cite{Hasna_Alouini_2003} have presented a framework for evaluating the end-to-end outage probability of a multi-hop wireless relay network with independent, non-regenerative relays, i.e., amplify-and-forward (AF), over Nakagami fading channels.  They provide a closed-form expression for the MGF of the reciprocal of the equivalent end-to-end SNR for independent Nakagami fading channels. 
For the same AF relay network,  Tsiftsis \cite{Tsiftsis08} obtains 
a closed-form  bound for the average error probability. This bound is reportedly tight at low SNR, but may become loose for higher SNR values and for more severe fading environment, e.g., Rayleigh fading. Amarasuriya,  Tellambura, and Ardakani \cite{Amarasuriya11}  present  an alternative bound to \cite{Hasna_Alouini_2003} on the end-to-end  SNR in   multi-hop AF relay networks. They derive the distribution function and the MGF for i.i.d.  Nakagami-m  fading and  
for independent, but non-identically distributed Rayleigh fading. 
The works above study physical layer performance bounds of channel-assisted, amplify-and-forward relaying over a multi-hop fading channels. They do not consider buffering or traffic burstiness, and are not concerned with 
network performance metrics addressed in this paper.  Delay 
and backlog analysis and optimization of multihop wireless networks 
remains an open research problem \cite{Le_Ekram08}.

\begin{figure}
\centering
\includegraphics [width=4.5in]{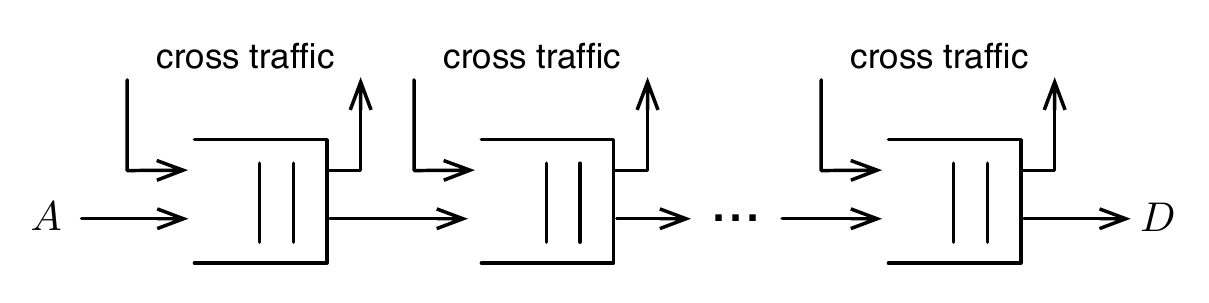}
\caption{Tandem network model.}
\label{fig:tandem}

\end{figure}

\section{Network Model in the Bit and SNR Domains}  \label{sec:model}

We consider a wireless $N$-node tandem network 
as shown in Fig.~\ref{fig:tandem}, where each node is modeled by 
a server with an infinite buffer. We are interested in the performance 
experienced by a (through) flow that traverses the 
entire network and may encounter cross traffic at each node. 
One can  think of the cross traffic   at a node as  the 
aggregate of all traffic traversing the node
that does not belong to the through flow. The service given 
to the through flow at a node is  a random process, which is 
governed by the instantaneous channel capacity as well as the cross traffic 
at the  node. We consider a fluid-flow traffic model where 
the flow is infinitely divisible. 
We will work in a discrete-time domain 
$\mathcal T = \{t_i: t_i =i \ \Delta t, i\in \mathds Z\}$, 
where $\mathds Z$ is the set of integers 
and   $\Delta t$ is length of the time unit. 
Setting  $\Delta t = 1$ allows us to replace $t_i$ by $i$, 
which we interpret as the index of a time slot.
We assume that the system is started with empty queues at time $t=0$.  

Different nodes and different traffic flows will be distinguished
by subscripts. The cumulative arrivals to, 
the service offered by, and the departures from
the node are represented by random processes 
$A_n$, $S_n$, and  $D_n$
that will be described more precisely below,
with $A_n=D_{n-1}$ for $n=1,\dots, N-1$.
We denote by $A=A_1$ and $D= D_N$
the arrivals to and the departures from the tandem network.
Throughout, we assume that arrival and service processes satisfy 
stationary bounds.

\subsection{Traffic and Service in the Bit Domain}

Consider for the moment a single node.
Dropping subscripts, we write
$$A(\tau,t) = \sum_{i=\tau}^{t-1}  a_i\,, 
\quad \text{and} \quad 
D(\tau,t) = \sum_{i=\tau}^{t-1} d_i\, ,
$$
for the cumulative arrivals and departures, respectively, at the node 
in the time interval $[\tau, t)$, 
where $a_i$ denotes the arrivals and $d_i$ the departures 
in the $i$-th time slot.
Due to causality, we have $D(0,t) \le A(0,t)$.
The  processes lie
in the set $\F$ of  non-negative  
bivariate functions $f(\tau,t)$ that are increasing in the 
second argument and vanish 
unless $0\le\tau<t$.  
The backlog at time $t >0$ is given by
\begin{equation} \label{eq:backlog}
B(t) = A(0,t)- D(0,t) \ , 
\end{equation}
and the delay  at the node is given by  
\begin{equation} \label{eq:delay}
W(t) = \inf \left\{u\ge 0: A(0,t)\le D(0,t+u)\right\} \ . 
\end{equation}

The service of the node in the time interval 
$[\tau, t)$ is given by a random process $S(\tau, t)$, 
such that Eq.~\eqref{eq:conv} holds for 
every arrival process $A$ and the
corresponding departure process $D$.  
This service description with bivariate functions 
is referred to as {\it dynamic server}. 
Initially defined for non-random service~\cite{ChangCruz99}, 
dynamic servers have been extended
to random processes in \cite{CSChang,FidlerMGF}. 

The above model is a typical network-layer model. 
where traffic is measured in bits, and service 
is measured in bits per second. We thus refer to this model 
of arrivals, departures, and service as residing in a {\it  bit domain}.

The network calculus exploits that networks  
which satisfy  the input-output relation of  Eq.~\eqref{eq:conv} with equality 
can be viewed as linear systems in a \minplus dioid algebra 
\cite{leBoudec}.   
In the ($\mathds R \cup \{+\infty \}, \min,+$) 
dioid,  the minimum and addition take the place of the 
standard addition and multiplication operations.
The network calculus is based on the fact
that $(\F, \min,\ast)$ is again a 
dioid~\cite{CSChang}. Note that the min-plus convolution, which 
provides the second operation in the dioid, is not commutative
in $\F$.


\subsection{Service Model of Wireless Channel}
\label{subsec:service-snr}
To compute a service model for a wireless channel, we 
assume that the channel state information is sampled 
at equal time intervals $\Delta t$. With $\Delta t=1$, 
let $\gamma_i$ denote  the instantaneous signal-to-noise 
ratio observed at the receiver 
in the $i$-th sampling epoch. Then, $\gamma_i$ is a nonnegative
random variable that has the probability distribution of the 
underlying fading model. 
We assume that the random variables
$\gamma_i$  are independent and identically distributed. 
This assumption is justified when 
$\Delta t$ is longer than the channel coherence time. Otherwise, the assumption 
will give optimistic bounds. We emphasize that the network calculus  
in this paper applies to settings without independence, however, 
the derivation of performance bounds will proceed differently. 
Using Eq.~\eqref{eq:capacity},
the instantaneous service offered by the channel in the 
$i$-th slot is given by $\log g(\gamma_i)$.
and the corresponding service process is given by    
\begin{equation} \label{eq:cumulative_capacity}
    S(\tau,t) = \sum_{i=\tau}^{t-1}  \log g(\gamma_i)    \,,
\end{equation} 
where we haven chosen units such that
the constant in Eq.~\eqref{eq:capacity} takes the value $c=1$.

The service description in Eq.~\eqref{eq:cumulative_capacity}
requires us to work with the logarithm of 
fading distributions,
which presents a non-trivial technical difficulty
via the usual network calculus or queueing theory.
On the other hand, observe that the exponential
$\S(\tau,t) = e^{ S(\tau,t)}$ is described more simply by
\begin{equation} \label{eq:cumulative_ServiceCapacity}
\S(\tau,t) = \prod_{i=\tau}^{t-1}   g(\gamma_{i}) \ .
\end{equation} 
This  motivates the development of a system model 
that allows us to exploit the more tractable service representation 
in Eq.~\eqref{eq:cumulative_ServiceCapacity}.
In this alternative model, arrivals, departures, and service  reside 
in a different domain, where we can work directly with 
the distribution functions of the fading channel gain and 
the corresponding SNR at the receiver.  

\begin{figure}
\centering
\includegraphics [width=4.5in]{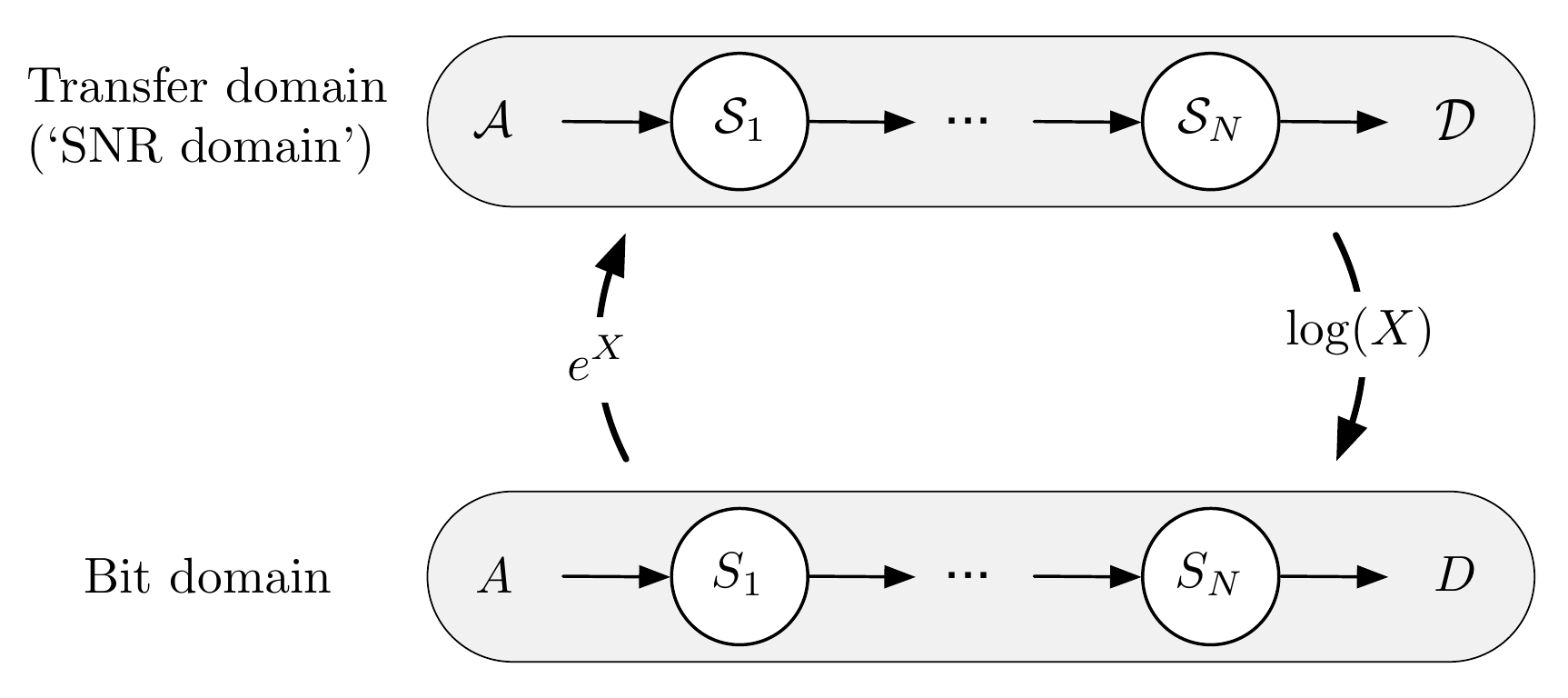}
\caption{Transfer Domain of Network Model.}
\label{fig:domains}
\end{figure}

\subsection{Network Model in the SNR Domain} \label{subsec:traffic-snr}

We now proceed by mapping the network model from 
Fig.~\ref{fig:tandem} into a transfer domain,  which we refer
to as  {\it SNR domain}. We will seek to 
derive performance bounds in the  transfer domain, and 
then map the results to the bit domain to obtain network-layer 
bounds for backlog and delays. 
The relationship of the network models in bit domain and SNR domain  
is illustrated in Fig.~\ref{fig:domains}.  

In the previous subsection, we constructed
the service process for a wireless link in 
the SNR domain in Eq.~\eqref{eq:cumulative_ServiceCapacity} 
as 
$$
\S(\tau,t)=e^{S(\tau,t)} \,.
$$
By analogy, we describe the arrivals and departures in the SNR domain by
$$
\A(\tau,t)\deq e^{A(\tau,t)}\,,\quad \text{and }\quad \D(\tau,t)\deq e^{D(\tau,t)}\,.
$$
Throughout  this paper, we use calligraphic upper-case 
letters to represent processes that characterize 
traffic or service as a function of the instantaneous SNR 
in the sense of  Eq.~\eqref{eq:cumulative_ServiceCapacity}. 
Due to the monotonicity of the exponential function, 
$\D(0,t)$ and $ \A(0,t)$ are increasing  
in $t$, and satisfy the causality property $\D(0,t) \le \A(0,t)$. 
The backlog process is accordingly described by
$$
\B(t)\deq e^{B(t)}=\A(t)/\D(t)\,.
$$
Since time is not affected by this transformation,
the delay is given by
\begin{align}
\label{eq:def-delay}
\W(t)\deq W(t) = \inf\{u\ge 0: \A(t)\le \D(t+u)\}\,.
\end{align}
To interpret these processes in the transfer domain, 
let $\gamma_{a,i}\deq g^{-1}(e^{a_i})$ be the instantaneous channel 
SNR required to transmit $a_i$ in a single time slot,
assuming transmission at the rate of the capacity limit.
The arrival process in the SNR domain can then be expressed
in terms of these variables as
\begin{equation} \label{eq:cumulative_arrivalCapacity}
\A(\tau,t) = \prod_{i=\tau}^{t-1} g(\gamma_{a,i})\,.
\end{equation} 
Here, we are treating channel quality expressed in terms 
of the instantaneous SNR as a commodity. An arrival in a time 
unit represents a 
workload, where  $\gamma_{a,i}$ expresses the amount of 
resources that will be consumed by the workload. 
The backlog can  similarly be expressed in terms of
the instantaneous SNR as
$$ \B( t) = \prod_{i=t}^{t+\tau_B-1} g(\gamma_{i}) \ , $$  
with the interpretation that  a node with backlog 
$B(t)$ at time $t$  requires full use of the channel 
capacity for $\tau_B$ time units to clear the backlog.   


Most importantly, the concept of the dynamic server translates to the 
SNR domain. In a network system, the
service process in the bit domain satisfies
Eq.~\eqref{eq:conv} if and only if
the process $\S(\tau,t)=e^{S(\tau,t)}$ in the SNR domain
satisfies
\begin{align}\label{eq:dynamicserver}
        \D(0,t) \ge  \inf_{0\le u \le t} \{ \A(0,u) \cdot \S (u,t)  \}  \ .
\end{align}
We refer to a network element that satisfies Eq.~\eqref{eq:dynamicserver} 
for any sample path as {\it dynamic SNR server}.
In this general setting, we not require that $\S$ takes the form in 
Eq.~\eqref{eq:cumulative_ServiceCapacity}, in particular, $\S(\tau,t)$  
does not be equal to $\S(\tau,u)\cdot\S(u,t)$.

Traffic aggregation in the SNR domain is  expressed in 
terms of a product. When $M$ flows have arrivals at a 
node with arrival processes denoted by $A_{k}, k=1,\ldots, M$, 
then the total arrival, $A_{\rm agg}$, are given for any $ 0 \leq \tau \leq t$ by 
$$
A_{\rm agg} (\tau,t) =  \sum_{k=1}^{M} A_k(\tau,t)  \ .
$$
If we 
let  $\A_{j}$ and $\A_{\rm agg}$ denote the 
corresponding processes in the SNR domain, we see that  
\begin{align*}
\A_{\rm agg} (\tau,t) = \prod_{k=1}^M \A_{k}(\tau,t)  \ .
\end{align*}
 
With the above definitions, the usual network description 
by a \minplus dioid algebra in 
the bit domain can be expressed in the SNR domain by 
a dioid algebra on $\F$
where the second operator is a multiplication. 
This enables the development of the 
\mx~network calculus in Sec.~\ref{sec:calculus}. 
We observe that the exponential function
defines a one-to-one correspondence between 
arrival and departure processes in the bit and SNR domains. 
The physical arrival, departure, service, and backlog
processes can be recovered from their counterparts in the
SNR domain by taking a logarithm (see Fig.~\ref{fig:domains}).

\section{Stochastic \mx~~Network Calculus}  \label{sec:calculus}

This section contains our main contribution: an analytical 
framework  for statistical end-to-end performance bounds for a network, 
where 
service is expressed  in terms of fading distributions 
residing in the SNR domain.

By an {\it SNR process} we mean a bivariate process 
$\X(\tau,t)$ taking values in $\mathds R^+$
that is increasing in the 
second argument, with $\X(t,t)=1$
for all $t$.  The space of SNR processes will be denoted by $\F^+$.
For any pair of SNR processes 
$\X(\tau,t)$ and $\Y(\tau,t)$, set
\begin{equation} \label{eq:convolution}
\X \conv \Y (\tau,t) \deq 
\inf_{\tau\leq u \le t} \big \{ \X (\tau,u) \cdot \Y (u,t) \big \} \ , 
\end{equation} 
and
\begin{equation} \label{eq:deconvolution}
\X \deconv \Y (\tau,t) 
\deq \sup_{ u \leq \tau} \Big \{ \frac{\X 
(u, t)}{\Y(u,\tau)} \Big \} \,.
\end{equation} 
We refer to `$\conv$' and `$\deconv$' as the the \mx~convolution and
\mx~deconvolution operators, respectively. 

The arrival, departure, and service processes constructed in
the previous section are SNR processes.
With the \mx~convolution, we can express the defining
property of a dynamic SNR server from
Eq.~\eqref{eq:dynamicserver} as
\begin{align} \label{eq:conv-mx}
D(0,t) \geq A \conv S (0,t) \  
\end{align}
for every pair of SNR arrival and departure processes $\A(\tau,t)$
and $\D(\tau,t)$.

\subsection{\mx~Dioid Algebras} 
We note that, in fact, for any system  description in the bit domain by   
the  ($\mathds R \cup \{+\infty \}, \min,+$) and  the  
(${\F}, \min,\ast$) dioid algebras there exists a corresponding  
characterization in the SNR domain using    
($\mathds R^+ \cup \{+\infty \}, \min,\times$) and 
(${\F^+}, \min,\conv$)~dioids. 
The following argument confirm that the properties of a dioid are satisfied. 

\bigskip
\begin{lemma} 
\label{lm:mx-dioid}
$(\mathds R^+ \cup \{+\infty \}, \min,\times)$ is a dioid. 
\end{lemma}

\bigskip
\begin{proof}
We show that ($\mathds R^+ \cup \{+\infty \}, \min,\times$) satisfies the dioid axioms. 
For  $a,b,c \in \{\mathds R^+ \cup \{+\infty \} \}$:

\bigskip
\begin{enumerate}
\item[(1)] \textit{Commutativity of $\min$:} $\min (a,b) =\min (b,a)$.
\item[(2)] \textit{Associativity of $\min$:} 
$\min (\min (a,b), c) =\min(a, \min (b,c))$.
\item[(3)] \textit{Idempotency of $\min$:} $\min (a,a) =a$.
\item[(4)] \textit{Associativity of $\times$:} $ (a\times b)  \times c= a \times  (b\times c)$.
\item[(5)] \textit{Distributivity of $\times$:} 
$ \min (a, b)  \times c= \min (a \times c,  b \times c)$.
\item[(6)] \textit{Existence of a null element:} The null element is $+\infty$. 
since $\min (+\infty, a) =a$.
\item[(7)] \textit{Absorption of the null element:}\ $ (+\infty) \times a =a \times (+\infty) =  +\infty $.
\item[(8)] \textit{Existence of a unity element:} The unit of multiplication is 
1, since $ 1 \times a =a \times 1 = a$.

\end{enumerate}
\end{proof}

\begin{lemma} 
$(\F^+, \min,\conv)$ is a dioid. 
\end{lemma}

\bigskip
\begin{proof}
Given bivariate functions $\X, \Y, \Z \in \F^+$. 
Since the $\min$ operation is a pointwise minimum,  properties 
of the $\min$ operation, that is, properties (1)--(3) from the proof of Lemma~\ref{lm:mx-dioid}, follow from the $(\mathds R^+ \cup \{+\infty \}, \min,\times)$ dioid.  
For the remaining properties we have 

\bigskip
\begin{itemize}
\item[(4)] \textit{Associativity of $\conv$:} 
\begin{align*}
(\X\conv \Y)  \conv \Z (\tau, t)  &  \\
& \hspace{-2cm} = \inf_{\tau\leq u \le t} \bigl\{ 
\inf_{\tau\leq s \le u} \bigl\{
\X (\tau,s) \cdot \Y (s,u) \bigr\} 
\cdot  \Z (u,t)  \bigr\} \\
& \hspace{-2cm} = \inf_{\tau\leq s \le u \le t} \bigl\{ 
\X (\tau,s) \cdot \Y (s,u)  \cdot  \Z (u,t)  \bigr\} \\
&  \hspace{-2cm}= 
\inf_{\tau\leq s \le t} \bigl\{ 
\X (\tau,s)  \cdot \inf_{s\leq u \le t} \bigl\{ \Y (s,u) 
\cdot  \Z (u,t) \bigr\}  \bigr\} \\
&  \hspace{-2cm}= \X\conv (\Y  \conv \Z) (\tau, t)  \ . 
\end{align*}
\item[(5)] \textit{Distributivity of $\conv$ over finite sums:} 
\begin{align*}
\min (\X, \Y)  \conv \Z (\tau, t) & \\
& \hspace{-2cm} = \inf_{\tau\leq s \le t} \bigl\{
\min (\X, \Y) (\tau,s) \cdot \Z (s,t) \bigr\} \\
& \hspace{-2cm} = \inf_{\tau\leq s \le t} \bigl\{
\min \bigl( \X (\tau,s) \cdot \Z (s,t) , \Y (\tau,s) \cdot \Z (s,t) \bigr)\bigr\} \\
& \hspace{-2cm} = \min \left( \X\conv \Y (\tau, t) , \Y  \conv \Z (\tau, t)\right) \, .
\end{align*}
\item[(6)] \textit{Existence of a null element:} The null element is $N(\tau, t)=+\infty$ for all values of $\tau$ and $t$.

\item[(7)] \textit{Absorption of the null element:} 
\begin{align*}
N \conv \X (\tau, t) = \inf_{\tau\leq u \le t} \big \{ (+\infty) \cdot \X (u,t) \big \} = +\infty \, . 
\end{align*}
Note that functions in $\F^+$ are strictly positive by definition.  

\item[(8)] \textit{Existence of a unity element:} The unity element 
is $\Delta(\tau, t)$, where 
\[
\Delta(\tau,t) = 
\begin{cases}
1 & \tau \geq t \, , \\
\infty & \tau < t \ . 
\end{cases}
\]
This gives 
\begin{align*}
\Delta  \conv \X (\tau, t) &  = \inf_{\tau\leq s \le t} \bigl\{
\Delta (\tau,s) \cdot \X (s,t) \bigr\} \\
&  = 
\Delta (\tau,\tau) \cdot \X (\tau,t)  \\
&  = 
 \X (\tau,t)  \, . 
 \end{align*}
\end{itemize}
\end{proof}

\subsection{Server Concatenation and Performance Bounds} 
The existing network calculus in the bit domain 
allows for the concatenation of tandem service elements using the 
\minplus~convolution (see page~1).
As an immediate consequence, single node performance bounds 
are extended to a multi-hop setting. 
We now show establish the corresponding result in 
the \mx~network calculus. Specifically, the concatenation of dynamic SNR servers is 
again a dynamic SNR server. We will prove the result for a tandem network of two nodes, as shown in Fig.~\ref{fig:concat}. 



\begin{lemma}\label{thm:conc}
Let $\S_1 (\tau,t)$ and $\S_2 (\tau,t)$ be two dynamic SNR servers in tandem as 
shown in Fig.~\ref{fig:concat}. Then, the service offered by the tandem of 
nodes   
is given by the dynamic SNR server $\S_{\rm net}(\tau,t)$ with 
\begin{equation*}
\S_{\rm net}(\tau,t) = \S_1 \conv \S_2(\tau,t) \ .  
\end{equation*}
\end{lemma}

\begin{proof}
Using Eq.~\eqref{eq:dynamicserver}, the departure process $\D(0,t) $ can be written as 
 \begin{eqnarray*}
 \D(0,t) &\ge& \inf_{0\le u\le t} \{ \A_2(0,u) \cdot \S_2 (u,t) \} \nonumber \\
 & \ge&  \inf_{0\le u \le t} \big\{  \inf_{0\le \tau \le u} \{ \A(0,\tau) \cdot  \S_1 (\tau,u) \} \cdot \S_2 (u,t)  \big  \} \nonumber \\
  & = & \!\!\!\! \inf_{0\le \tau \le t} \big\{ \A(0,\tau) \cdot \!\!  \inf_{\tau \le u \le t} \{ \S_1 (\tau,u)   \cdot \S_2 (u,t) \} \big  \} \\
  & = & \!\!\!\! \inf_{0\le \tau \le t} \big\{ \A(0,\tau) \cdot (\S_1 \conv \S_2) (\tau,t) \big  \} \,. \\[-0.5cm]
 \end{eqnarray*}
\end{proof}

The extension to  networks with more than two nodes follows by iteratively applying  Lemma~\ref{thm:conc}. Hence, the  dynamic network SNR server with $N$~nodes in tandem is given 
by 
\begin{equation}\label{eq:concatenationN}
\S_{\rm net}(\tau,t) =   \S_1  \conv \S_2  \conv \cdots \conv \S_N (\tau,t) \ . 
\end{equation}

\begin{figure}[t]
\centering
\includegraphics [width=4in]{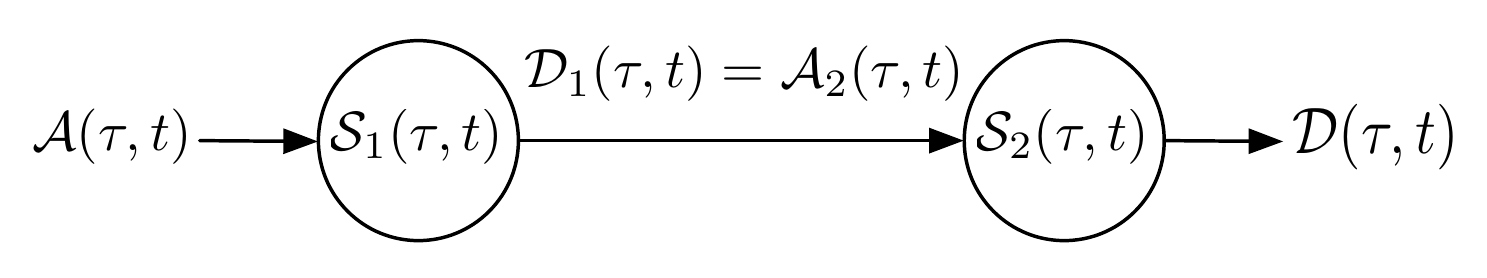}
\caption{Tandem of dynamic SNR servers.}
\label{fig:concat}
\end{figure}

Performance bounds in the \mx~network calculus 
are computed with the  \mx~deconvolution 
operator. This is analogous to role of the \minplus~deconvolution 
in the existing \minplus~network calculus. 
The bounds are expressed in the following lemma. 

%
%
%
%
%

\begin{lemma}\label{lem:OBbound}
Given a system with SNR arrival process $\A(\tau,t)$ and  dynamic SNR server $\S(\tau,t)$. 
\begin{itemize}
\item {\sc Output Burstiness.} 
The SNR departure process 
is bounded by $\D(\tau,t) \le \A \deconv \S (\tau, t)$.

 \item {\sc Backlog Bound.} 
 The SNR backlog process 
is bounded by $\B(t) \le \A \deconv \S (t, t)$.

 \item {\sc Delay Bound.} 
The  delay process 
is bounded by $ \W(t) \le 
 \inf \Big \{d \ge 0: \A \deconv \S (t+d,t)\le 1  \Big \}$. 
\end{itemize}
\end{lemma}

\begin{proof}
For the output bound, we fix $\tau$ and $t$ with $0 \le \tau \le t$ and derive 
\begin{eqnarray*} \label{eq:depart2}
 \D(\tau,t) &=& \frac{\D(0,t)}{ \D(0,\tau)} \le \frac{\A(0,t)}{ \D(0,\tau)}  \nonumber \\
 &\leq& \sup_{0 \leq u \le\tau  } \Big \{ \frac{\A(0,t)}{\A(0, u) \cdot \S(u,\tau) } \Big \} \nonumber \\
 &=&  \sup_{0 \leq u \le \tau  } \Big \{\frac{ \A(u,t)}{\S(u, \tau)} \Big \}  \, ,
\end{eqnarray*} 
where we used the inequality $\D(0,\tau) \ge  \A \conv \S(0,\tau)\}$ in the second line.  

For any fixed sample path, fix an arbitrary $t \geq 0$. 
The bound on the backlog is derived by 
\begin{eqnarray*} \label{eq:Backlog10}
 \B(t) &=& \frac{\A(0,t)}{ \D(0,t)} \nonumber \\
 &\leq& \sup_{0 \leq u \leq t} \Big \{ \frac{\A(0,t)}{\A(0,u) \cdot \S(u, t) } \Big \} \nonumber \\
 &=&  \sup_{0 \leq u \leq t} \Big \{\frac{ \A(u,t)}{\S(u, t)} \Big \} ,
\end{eqnarray*} 
where we used  $\D(0,t) \ge  \A \conv \S(0, t)\}$ in the second step. 

Recall that the delay is invariant under the transform of domains, that is, $\W (t) = W(t)$. By definition of the delay 
in Eq.~\eqref{eq:def-delay}, a delay bound $w$ satisfies
\begin{eqnarray} \label{eq:delayBound}
 \W(t) &=&  \inf \Big\{w \ge 0:   \frac{\A(0,t)}{\D(0,t+w)}\le 1 \Big\} \nonumber \\
 &\leq& \!\!\! \inf \Big \{w \ge 0: \!\!\!\sup_{0 \leq u \leq t} \Big \{ \frac{\A(0,t)}{\A(0,u) \cdot \S(u, t+w) } \Big \}\le 1  \Big \}  \nonumber \\
 &=&  \inf \Big \{w \ge 0: \sup_{0 \leq u \leq t} \Big \{ \frac{\A(u,t)}{  \S(u, t+w) } \Big \}\le 1  \Big \} .
\end{eqnarray} 
where we used the inequality $\D(0,t+w) \ge  \A \conv \S (0, t+w)\}$ in the second line.  
\end{proof}

With an algebraic description for network performance bounds in 
the SNR domain in hand, we now turn to the problem of computing the bounds. 

\subsection{The Mellin Transform in the SNR domain}
The concise (and familiar)  expressions 
from the previous section for the network service and performance bounds in
the SNR domain hide the difficulty of computing the expressions. 
In fact, all expressions of the \mx~network calculus 
contain products or quotients  of  random variables.  
The {\it Mellin transform} \cite{Davies} facilitates such computations,
particularly when the arrival and service processes 
are independent. 

The Mellin transform of a nonnegative random variable $X$ is
defined by
\begin{equation} \label{eq:Mt2}
  \M_{X} (s )= E[X^{s-1}]
\end{equation}   
for any complex number $s$ such that this expected value exists.


Among its many properties, we will exploit that the 
Mellin transform  of a product of two independent 
random variables $X$ and $Y$ equals the product of their  
Mellin transforms, 
\begin{equation} \label{eq:Property3MT}
 \M_{X\cdot Y} (s ) = \M_{X} (s ) \cdot \M_{Y} (s ) .
 \end{equation}
Similarly, the Mellin transform of the quotient of independent 
random variables is given by 
\begin{equation} \label{eq:Property4MT}
 \M_{X/Y} (s ) = \M_{X} (s ) \cdot \M_{Y} (2-s ) . 
 \end{equation}

We will evaluate the Mellin transform only for
$s\in\mathds R$, where it is always well-defined 
(though it may take the value $+\infty$). 
When $s>1$, the Mellin transform is order-preserving,
i.e., for any pair of random variables $X,Y$ with $\P(X>Y)=0$
we have $\M_X(s)\le \M_Y(s)$ for all $s$.
When $s<1$, the order is reversed.
Hence bounds on the distribution
of a random variable $X$ generally
imply bounds on its Mellin transform.

A more subtle question is how to obtain bounds on the
distribution of a random variable from its Mellin 
transform. Here, the complex inversion formula is not helpful.
Instead, we will use the moment bounds
\begin{align}
\label{eq:moment1}
\P (   X   \geq a ) &\le a^{-s}\M_X(1+s) 
\end{align} 
for all $a>0$ and $s>0$. 
For bivariate random processes $\X(\tau,t)$, we will
write $\M_{\X} (s,\tau,t)\deq \M_{\X (\tau,t)} (s)$.


In our calculus, we work with the Mellin transform
of \mx~convolutions and deconvolutions, 
which not only involves 
products and quotients, but also requires to compute infimums 
and supremums.  The computation of the exact Mellin transform 
for these operations is generally not feasible. 
We therefore resort to bounds, as specified  in 
the next lemma. 

\begin{lemma}\label{lem:MTConvDecon}
Let $\X(\tau,t)$ and 
$\Y(\tau,t)$ be two independent nonnegative bivariate
random processes.
For $s<1$, the
Mellin transform of the \mx~convolution $\X\conv\Y(\tau,t)$
is bounded by
 \begin{equation} \label{eq:convMT}
 \M_{ \X \conv  \Y} (s, \tau, t ) \le \sum_{u=\tau}^{ t}  
\M_{ \X} (s, \tau,u ) \cdot \M_\Y (s, u, t )\,.
\end{equation}  
For $s>1$, the Mellin transform of the \mx~deconvolution 
$\X\deconv\Y(\tau,t)$ is bounded by
\begin{equation} \label{eq:deconvMT}
 \M_{ \X \deconv  \Y} (s, \tau, t ) \le 
\sum_{  u =0}^{\tau}  
{\M_{ \X} (s, u,t ) }\cdot { \M_{ \Y} (2\!-\!s, 
u, \tau )} \,.
\end{equation}  
\end{lemma}

\begin{proof} Note that the function $f(z)= z^{s-1}$
is increasing for $s>1$ and decreasing for $s<1$.
For $s<1$, the convolution is estimated by 
\begin{align*} 
\M_{ \X \conv  \Y} (s, \tau, t ) 
& =    E \Bigl [ \Bigl (\ \inf_{\tau\le u \le t} 
\{  \X(\tau, u) \cdot  \Y (u,t) \} \Bigr)^{s-1} \Bigr]\\
& \hspace{-2cm}
= E \Bigl [\ \sup_{\tau\le u \le t} 
\{  (\X(\tau, u))^{s-1} \cdot  (\Y (u,t))^{s-1} \} 
\Bigr]\\
 & \hspace{-2cm}
\le \sum_{u=\tau}^{t} E \big[ (  \X(\tau, u)  )^{s-1} 
\big] \cdot E \big[  (    \Y (u,t)  )^{s-1} \big]\,.
\end{align*} 
In the last line, we have used the non-negativity of $\X$ and $\Y$
to replace the supremum with a sum, and 
their independence to evaluate the expectation of the products.
Eq.~\eqref{eq:convMT} follows by inserting the
definition of the Mellin transform.
The deconvolution is similarly estimated for $s>1$ by
\begin{align*} 
\M_{ \X \deconv  \Y} (s, \tau, t ) 
& = E \Bigl [ \Bigl (\ \sup_{ u \le \tau} \bigl\{  \X( u,t) /  
\Y (u,\tau) \bigr\} \bigr)^{s-1}\Bigr] \\
& \hspace{-2cm}
=    E \Bigl [\ \sup_{0  \le u \le \tau} 
\bigl\{ \bigl( \X( u,t)\bigr)^{s-1}
\cdot
\bigl(\Y (u,\tau) \bigr)^{1-s} \bigr\} \Bigr]\\
& \hspace{-2cm}
\le \sum_{u=0}^{\tau} 
E \bigl[ \bigl(  \X( u,t)\bigr)^{s-1}\bigr]
\cdot E\bigl[ \bigl(\Y (u,\tau)\bigr)^{1-s}\bigr]\,,
\end{align*}
and Eq.~\eqref{eq:deconvMT} follows 
from the definition of the Mellin transform. 
\end{proof}
As an application of Lemmas~\ref{thm:conc} and~\ref{lem:MTConvDecon},
we compute a bound on the
Mellin transform of the service process 
for a cascade of fading channels.
We also make the idealizing assumption that 
the channels are independent. 

\begin{corollary} \label{cor:cascade}
Consider a  cascade of $N$ independent,
identically distributed fading channels,
where the service process for
each channel is given by Eq.~\eqref{eq:cumulative_ServiceCapacity},
with i.i.d. random variables $\gamma_i$.
Let $\S_{\rm net}(\tau,t)$
denote the SNR service process for the entire cascade.
Then, the  Mellin transform for this process satisfies
\begin{align*}
\M_{\S_{\rm net}}(s,\tau,t) \le
\binom{N-1+t-\tau}{t-\tau} 
\cdot \bigl(\M_{g(\gamma)}(s) \bigr)^{t-\tau}
\end{align*}
for all $s<1$.
\end{corollary}

\begin{proof} We use the server concatenation formula
in Eq.~\eqref{eq:concatenationN} to 
represent the service of the cascade as
$\S_{\rm net}(\tau,t) =\S_1\conv \S_2\conv \dots\conv \S_N(\tau,t)$.
By Lemma~\ref{lem:MTConvDecon}, its Mellin transform satisfies
for $s<1$
\begin{align}\label{eq:Mbound}
\M_{\S_{\rm net}}(s,\tau,t) 
&\le \sum_{u_1\dots u_{N-1}} \prod_{n=1}^N M_\S(s,u_{n-1},u_n)\,,
\end{align}
where the  sum runs over all sequences $u_0\le u_1\le\dots \le u_N$
with  $u_0=\tau$ and $u_N=t$.
The assumptions on the service processes of the individual
channels imply that each product evaluates to 
the same function
\begin{align*}
\prod_{n=1}^N M_\S(s,u_{n-1},u_n) 
&= \prod_{n=1}^N \bigl(\M_{g(\gamma)}(s) \bigr)^{u_n-u_{n-1}}\\
&= \bigl(\M_{g(\gamma)}(s) \bigr)^{t-\tau}\,,
\end{align*}
where $\gamma$ is a random variable that has the same
distribution as the $\gamma_i$.
Since the number of summands in Eq.~\eqref{eq:Mbound}
is given by $\binom{N-1+t-\tau}{t-\tau}$, the claim follows.
\end{proof}

\subsection{Performance Bounds for the Bit Domain} \label{sec:bounds}

We next obtain  
network-level performance bounds for the bit domain.
This involves a transformation from the
SNR domain to the bit domain via the relationship in 
Fig.~\ref{fig:domains}, which provides the translation of 
the abstract metrics $\D$ and $\B$ into 
processes $D$ and $B$, which, along with $W$,
are concrete  measures for traffic burstiness, buffer requirements,
and delay.

\begin{theorem}\label{thm:BLbound}
Given a system where arrivals are described by a bivariate 
process $A(\tau,t)$, and the available service 
is given by a dynamic server $S(\tau,t)$. Let
$\A(\tau,t)$ and $\S(\tau,t)$ be
the corresponding SNR processes. Fix $\eps >0$ and define,
for $s>0$,
\begin{align*}
\Msum (s, \tau, t) = \sum_{u=0}^{\min(\tau,t)} \M_{\A} (1+s, u, t) 
 		  \cdot  \M_{\S} (1-s, u, \tau) \,.
\end{align*} 
Then, we have the following probabilistic performance bounds.

\noindent $\bullet$  {\sc Output Burstiness:} 
$\P \bigl( D(\tau,t) > d^\eps  \bigr) \leq \eps$,
where
\begin{align*}
d^\eps(\tau,t) = \inf_{s>0} 
   \Big\{ \frac{1}{s} \bigl( \log \Msum (s,\tau, t) 
-\log \eps \bigr) \Big \} \,;
\end{align*} 
\noindent $\bullet$  {\sc Backlog:} 
$\P \bigl( B(t) > b^\eps  \bigr) \leq \eps$, where
\begin{align*}
b^\eps = \inf_{s>0} 
   \Big\{ \frac{1}{s} \bigl( \log \Msum (s,t, t) 
-\log \eps \bigr) \Big \} \,;
\end{align*} 
\noindent $\bullet$  {\sc Delay:} 
$\P \bigl( W(t) > w^\eps  \bigr) \leq \eps$, where
$w^\eps$ is the smallest number satisfying 
\begin{align*}
\inf_{s>0} \Big\{ \Msum (s, t+w^\eps,t) \Big \} \leq \eps \,.
\end{align*} 
\end{theorem}


\begin{proof}
For the bound on the distribution of the
output burstiness, we start
from the inequality $\D(\tau,t)\le \A\conv\S(\tau,t)$.
It follows from the moment bound and Lemma~\ref{lem:MTConvDecon}
that, for any choice of $d>0$ and all $s>0$
\begin{align*}
\P(D(\tau,t)>d)&= \P(\D(t)> e^d)\\
 &\le \P(\A\deconv \S(\tau ,t)>e^d)\\
&\le (e^d)^{-s} \M_{\A\deconv\S}(1+s,\tau,t)\\
&= e^{-sd} \Msum (s,\tau,t)\,.
\end{align*}
To obtain the claim, we
set the right hand side equal to $\eps$, solve
for $d$, and optimize over the value of $s>0$
to obtain $d^\eps(\tau,t)$.
The proof of the backlog bound proceeds in the same way,
starting from the
inequality $\B(t) \le \A(0,t)/\D(0,t)$.

The delay bound is slightly more subtle.
Fix $t\ge 0$.
Using Lemma \ref{lem:OBbound}
and the moment bound with $a=1$, we obtain that
\begin{align*}
\P(\W(t)>w) &\le \P(\A\deconv\S(t+w,t) > 1)\\
&\le \M_{A\deconv\S}(s+1,t+w,t)\,
\end{align*}
for every $s>0$. By Lemma~\ref{lem:MTConvDecon},
the Mellin transform $\M_{A\deconv\S}(s+1,t+w,t)$
satisfies a bound that agrees with
the function $\Msum(s,t+w,t)$, except that
the upper limit in the summation that defines $\Msum(s,t+w,t)$ 
would have to be replaced by $\tau=t+w$.
To obtain a sharper estimate,
we use instead Eq.~\eqref{eq:delayBound}
from the proof of Lemma \ref{lem:OBbound}.
The resulting bound is that
$$
\Z(t)\deq \sup_{0\le u\le t}\Bigl\{\frac{\A(u,t)}{\S(u,t+w)}
\Bigr\}
$$
satisfies
\begin{align}
\notag
\P(\W(t)>w) &\le \P(\Z(t) > 1)\\
\notag &\le \M_{\Z(t)}(s+1)\\
\label{eq:Wbound}&\le \Msum(s,t+w,t)\,,
\end{align}
where we have used that the supremum in 
the definition of $\Z$ extends only
up to $u=t$, and then repeated
the proof Eq.~\eqref{eq:deconvMT}.
The claim follows by optimizing over $s$.
\end{proof}


Corresponding bounds as in Theorem~\ref{thm:BLbound} can be obtained 
using the \minplus algebra and the 
network calculus with moment-generating functions \cite{FidlerMGF}. 
The significance of Theorem~\ref{thm:BLbound} 
is that it permits the application of the the network calculus, 
where 
traffic is characterized in the bit domain, and service 
is naturally expressed in the SNR domain. 
This will become evident in the next section, 
where  the Theorem~gives us concise  
bounds for delays and backlog of multi-hop networks
with Rayleigh fading channels.

\section{Network Performance of Rayleigh Channels} \label{sec:rayleigh}

We now apply the techniques developed in the 
two previous sections to a network of Rayleigh channels.
Consider  the  dynamic SNR server description
for a Rayleigh fading channel, as constructed in Sec.~\ref{sec:model}.B.
We use Eq.~\eqref{eq:cumulative_ServiceCapacity},
with the function $g(\gamma)$ given by
\begin{equation} \label{eq:Ray1}
 g(\gamma) = 1+ \gamma = 1+ \bar{\gamma} |h|^2, 
\end{equation}
where $\bar{\gamma}$ is the average SNR of the 
channel and  $|h|$ is the fading gain. For Rayleigh fading, 
$|h|$ is a Rayleigh random variable with probability density
$f(x)= 2x e^{-x^2}$.
In a physical system, 
$\bar{\gamma}= \bar P_r / \sigma^2$, where $\bar P_r$ and $\sigma^2$ 
are the received signal power and the (additive white Gaussian) noise power at the receiver, 
respectively.  Then, $|h|^2$ is exponentially distributed, 
and the Mellin transform of $g(\gamma)$ is given by
\begin{align*}\label{eq:RayMT}
\M_{g(\gamma)}(s) = 
e^{ 1/\bar\gamma}   \bar{\gamma}^{s-1} \Gamma(s,\bar{\gamma}^{-1}) \, ,
\end{align*}
where $\Gamma(s, y)=\int_y^\infty x^{s-1}e^{-x}\,dx$ 
is the upper incomplete Gamma function.
Using the assumption that the $\gamma_i$ are independent and identically
distributed, we obtain for the Mellin transform 
of the dynamic server
\begin{equation}\label{eq:RayMT1}
\M_{\S}(s, \tau, t)=\Bigl(e^{ 1/{  \bar{\gamma}}}  
\bar{\gamma}^{s-1} \Gamma(s,\bar\gamma^{-1}) \Bigr)^{t-\tau}.
\end{equation}

For the arrival process, we use a  characterization due
to Chang \cite{CSChang}, where 
the moment-generating function of the cumulative arrival 
process in the bit domain is bounded by 
\begin{align*} 
\frac{1}{s} \log E [e^{s A(\tau,t) }] \le  \rho(s)\cdot(t-\tau) +\sigma(s)
\end{align*}
for some $s>0$. In general, $\rho(s)$ and $\sigma(s)$
are nonnegative increasing functions of $s$
that may become infinite when $s$ is large.  
This traffic class,  referred to 
as \sr~bounded arrivals, is broad enough 
to include Markov-modulated arrival processes.  
The corresponding class of SNR arrival processes is defined
by the condition that
\begin{align}\label{eq:srMA}
\M_{\A}(s,\tau,t)    \le    e^{(s-1) \cdot (\rho(s-1)\cdot(t-\tau) 
+ \sigma(s-1))}\ 
\end{align}
for some $s>1$.

\subsection{Performance Bounds of Rayleigh Fading Channels} 
We consider the transmission of \sr~bounded arrivals  on 
a Rayleigh fading channel. To obtain single-hop
performance bounds, we apply
Theorem~\ref{thm:BLbound} with the expressions for the
 Mellin transforms for the SNR service and arrival
processes from Eqs.~\eqref{eq:RayMT1}
and \eqref{eq:srMA}. For the function $\Msum (s,\tau,t)$
from the statement of the theorem, we compute 
\begin{align*}
\Msum (s, \tau, t) & 
\leq \ e^{s(\rho(s)(t-\tau)+\sigma(s))} 
\sum_{u = [\tau-t]_+}^{ \infty}  
\big (  \underbrace{e^{s  \rho(s)}   e^{1/\bar\gamma} 
\bar{\gamma}^{-s} \Gamma(1-s,\bar\gamma^{-1})}_{\deq V(s)} \big )^{ u} \,,
\end{align*}
where $[\tau-t]_+$ is the maximum of $\tau-t$ and $0$.
The sum converges when $V(s)<1$, which can be interpreted as a
stability condition.  
Inserting the result into Theorem~\ref{thm:BLbound}, we obtain
for the output burstiness the probabilistic bound
\begin{align*}
d^\eps_{\rm net}(\tau,t)  & =  
   \inf_{s> 0} \Bigl\{ \rho(s)(t-\tau)+
\sigma(s) 
-\frac{1}{s} \bigl(\log (1\!-\!V(s))  
+   \log \eps \bigr)    \Bigr\} \,.  
\end{align*}
The backlog bound is obtained by setting $\tau=t$,
 \begin{align*}
 b^\eps    =  
  \inf_{s> 0} \left\{ \sigma(s)  -\frac{1}{s}
\bigl (\log (1-V(s)) + \log \eps \bigr) \right\}\,.  
 \end{align*}   
The delay bound is the smallest number $w^\eps$ such that
\begin{align*}
\inf_{s>0}\Bigl\{
\frac{e^{s\cdot(-\rho(s)w^\eps+\sigma(s))} }{1-V(s)} 
\cdot (V(s))^{w^\eps}\Bigr\}\le \eps\,.
\end{align*}

We also derive end-to-end bounds for 
a cascade of $N$ Rayleigh channels
with \sr~bounded arrivals, using the same parameters as before. 
Let $\S_{\rm net}(\tau,t)$ be the service
process for the entire cascade.
By Corollary~\ref{cor:cascade}, its Mellin 
transform satisfies for $0\le \tau\le t$
\begin{align*}
\M_{\S_{\rm net}}(s,\tau,t) 
\le 
\binom{N-1+t-\tau}{t-\tau} 
\cdot \Bigl(e^{ 1/{  \bar{\gamma}}}  
\bar{\gamma}^{s-1} \Gamma(s,\bar\gamma^{-1}) \Bigr)^{t-\tau}\!.
\end{align*}
We will use again Theorem~\ref{thm:BLbound}.
For $0\le \tau\le t$, we compute 
\begin{align*}
\Msum_{\rm net} (s, \tau, t) 
&\leq \ \frac{e^{s\cdot(\rho(s)(t-\tau)+\sigma(s))}}{(1-V(s) )^N}
\,,
\end{align*}
where $V(s)$ is as defined above, and where we applied the 
combinatorial identity
\begin{align}\label{eq:MultiNodePerf}
\sum_{u=0}^{\infty}  x^u \binom{N-1+u}{u}  = \frac{1}{(1-x)^N} \ ,  
\end{align}
for any $N > 1$ and for $0< x <1$. 
Inserting $\Msum_{\rm net} (s, \tau, t)$ into  
Theorem~\ref{thm:BLbound}
gives for the end-to-end output bound, 
denoted by $d^\eps_{\rm net}(\tau,t)$,  the value
\begin{align*}
d^\eps_{\rm net}(\tau,t)  & =  
   \inf_{s> 0} \Bigl\{ \rho(s)(t-\tau)+
\sigma(s) 
-\frac{1}{s} \bigl( N 
\log (1\!-\!V(s))  +   \log \eps \bigr)    \Bigr\} \,.  
\end{align*} 
Note that for $N=1$, this bound agrees
with the previous bound for a single node.
In the same way, we derive the probabilistic backlog bound
\begin{align}\label{eq:MultiNodeBacklog}
b^\eps_{\rm net}  & =  
   \inf_{s> 0} \left\{ \sigma(s)  -\frac{1}{s} \bigl( N 
\log (1-V(s))  +   \log \eps \bigr)    \right\} \,.  
\end{align} 
For the delay bound, we estimate for $w\ge 0$
\begin{align}
\notag
\Msum_{\rm net} (s, r+w, t) 
&\leq \ e^{-s(\rho(s) w +\sigma(s))}
\sum_{u=w}^\infty \binom{N-1+u}{u}(V(s))^u\\
\label{eq:delay-M}
&\leq \ 
\frac{e^{s(-\rho(s) w +\sigma(s))}}{(1-V(s) )^N} 
\cdot\min \left\{1,(V(s))^w w^{N-1} \right\} \,.
\end{align}
Here, the first term in the minimum is obtained
by extending the summation
down to $u=0$, and the second term results
from the inequality
\begin{align*}
\binom{N-1+u}{u}\le w^{N-1}\cdot \binom{N-1+u-w}{u-w}
\end{align*}
for $u\ge w$. In both cases, the resulting sum
can be evaluated with Eq.~\eqref{eq:MultiNodePerf}.
The delay bound $w^\eps$ is determined according
to Theorem~\ref{thm:BLbound} by setting
the right hand side of Eq.~\eqref{eq:delay-M}
equal to $\eps$, solving for $w$, and minimizing over $s$.
Because of the complexity of the bound
in Eq.~\eqref{eq:delay-M}, the last two steps
can only be performed numerically.

It is apparent that the complexity of computing end-to-end 
bounds is no different than bounds for a single channel. 
More importantly, observe that the end-to-end  bounds scale 
linearly in  the number of nodes $N$. Relaxing 
the strong independence assumptions on the channel
properties would give different scaling properties.

\begin{figure}
\centering
\includegraphics [width=5in]{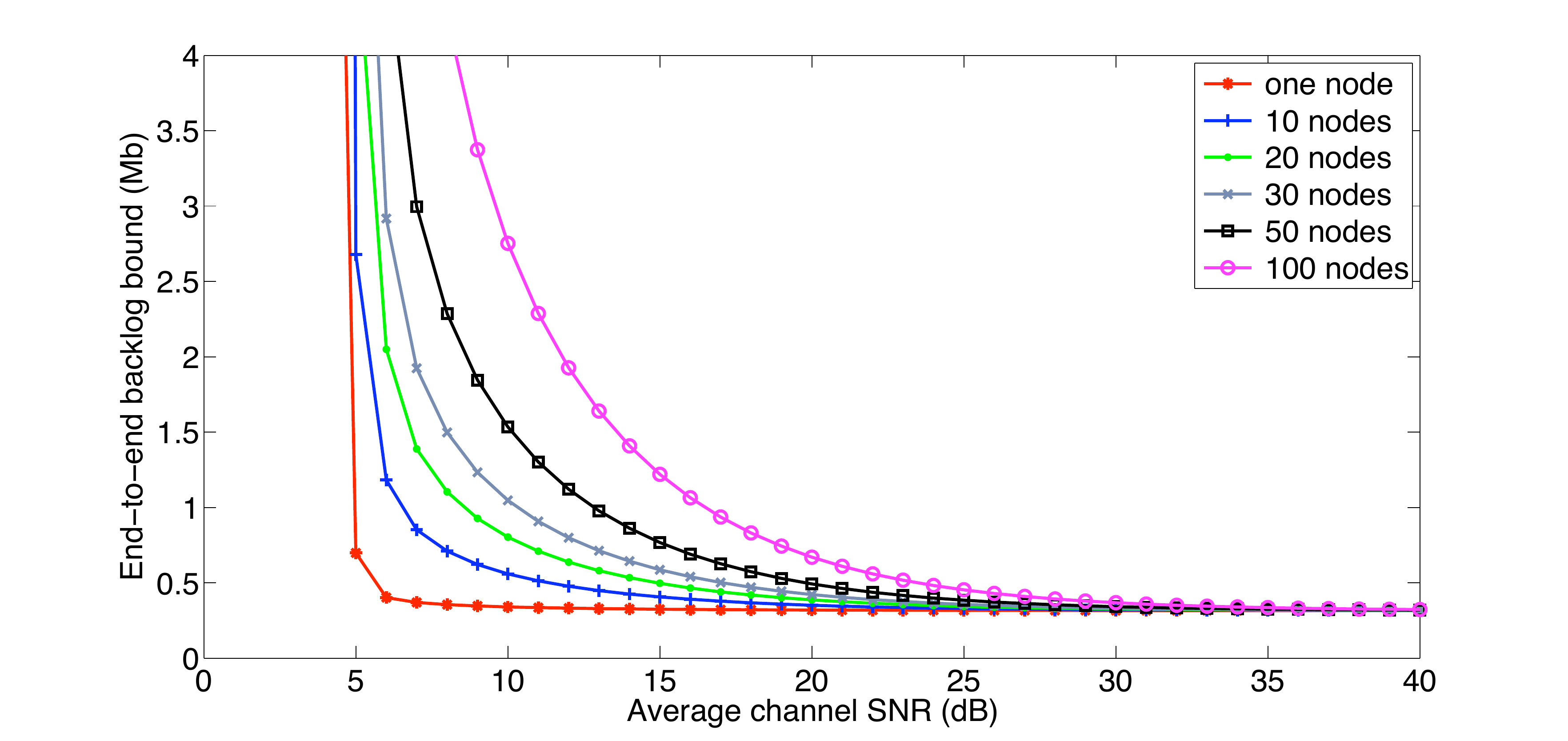}
\caption{End-to-end backlog bound ($b^\eps_{\rm net}$) vs. average channel SNR ($\bar{\gamma} $) for multihop Rayleigh fading channels with  $\eps =10^{-4}$,  \sr~bounded traffic with $\sigma(s) = 50$~kb and $\rho(s) = 30$~kbps, and $W=20$~kHz.}
\label{fig:backlog-vs-SNR-multihop}

\vspace{-4mm}
\end{figure}
\begin{figure}
\centering
\includegraphics [width=5in]{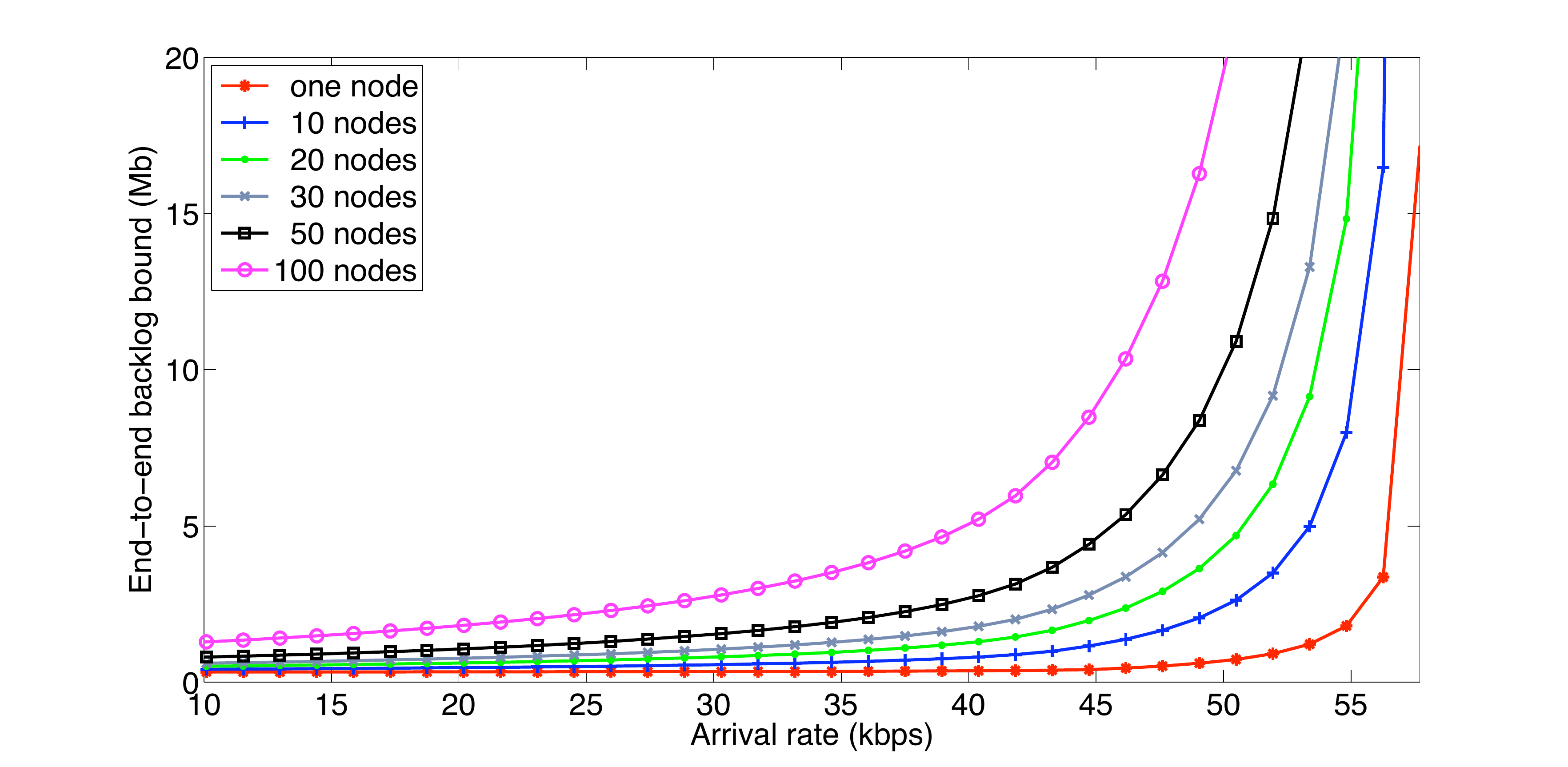}
\caption{End-to-end backlog bound  ($b_{net}^{\eps}$) vs. arrival rate ($\rho(s)$) for   multi-hop   Rayleigh fading channels with   $\eps =10^{-4}$, \sr ~bounded traffic with $\sigma(s) = 50$~kb and $\bar{\gamma} = 10$~dB, and $W=20$~kHz.}
\label{fig:backlog-vs-rho-multihop}

\centering
\includegraphics [width=5in]{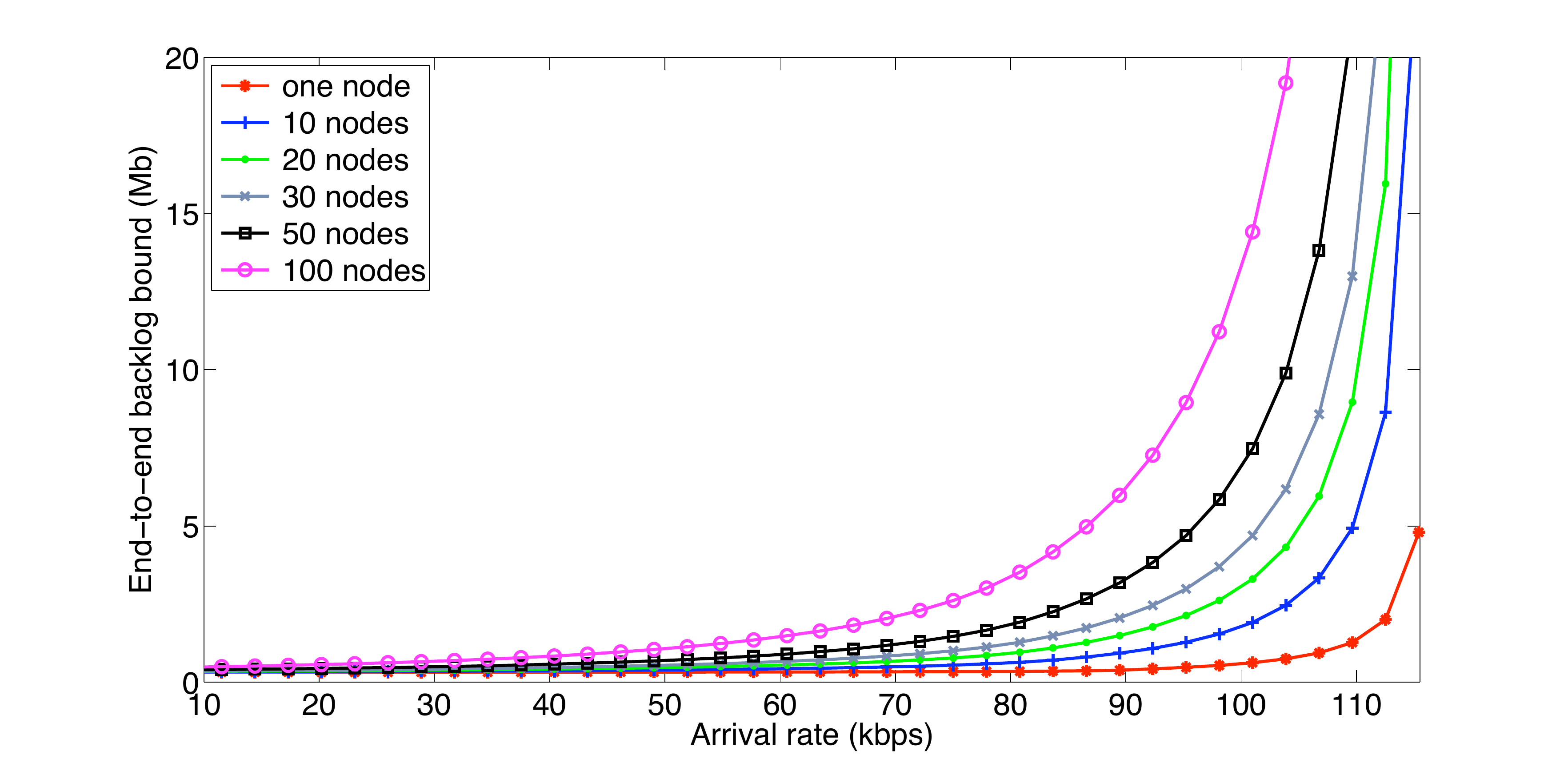}
\caption{End-to-end backlog bound ($b^\eps_{\rm net}$) vs. arrival rate ($\rho(s)$) for multi-hop  Rayleigh fading channels with $\eps =10^{-4}$, \sr~bounded traffic with $\sigma(s) = 50$~kb, $\bar{\gamma} = 20$~dB, and $W=20$~kHz.}
\label{fig:backlog-vs-rho2-multihop}

\vspace{-4mm}
\end{figure}

\subsection{Numerical Examples}
\label{sec:numerical}
We next present numerical results, where we assume Rayleigh channels with a transmission bandwidth of $W =20$~kHz. 
For traffic, we use \sr~bounded arrivals with default values $\sigma(s) \equiv 50$~kb  and $\rho(s) \equiv 30$~kbps, i.e., the bounds on rate and bursts are deterministic. Hence, the source of randomness in 
the examples results only from randomness of the channels.
We use a violation probability of $\eps = 10^{-4}$.

In Fig.~\ref{fig:backlog-vs-SNR-multihop} we show the end-to-end backlog for cascades  
of $N$~Rayleigh channels, as a function of the average SNR of each channel. 
Even though the backlog bounds increase only linearly in the number of nodes, 
the per-node requirements -- at least for the last node of the cascade -- must 
satisfy the end-to-end bounds, since it cannot be assumed that  backlog 
is equally distributed across the nodes. 
When the SNR of the nodes is sufficiently high, 
the backlog remains low even for a large number of hops. We observe 
that the channel becomes saturated for $\bar{\gamma}=5$~dB. When 
the number of channels is small, the backlog increases sharply in the vicinity of  $\bar{\gamma}=5$, but remains low everywhere else.

In Figs.~\ref{fig:backlog-vs-rho-multihop} and~\ref{fig:backlog-vs-rho2-multihop} we present, for  values of the 
average SNR of  $\bar{\gamma}=10$~dB (Fig.~\ref{fig:backlog-vs-rho-multihop}) and $\bar{\gamma}=20$~dB (Fig.~\ref{fig:backlog-vs-rho2-multihop}), how the end-to-end backlog 
increases as a function of the transmission rate, for 
different network sizes. We observe that the maximum achievable rate  
that does not result in a `blow-up' of the backlog decreases as the 
number of nodes is increased.  We observe that the blow-up occurs 
earlier when the channel has a lower SNR.

\begin{figure}
\centering
\includegraphics [width=5in]{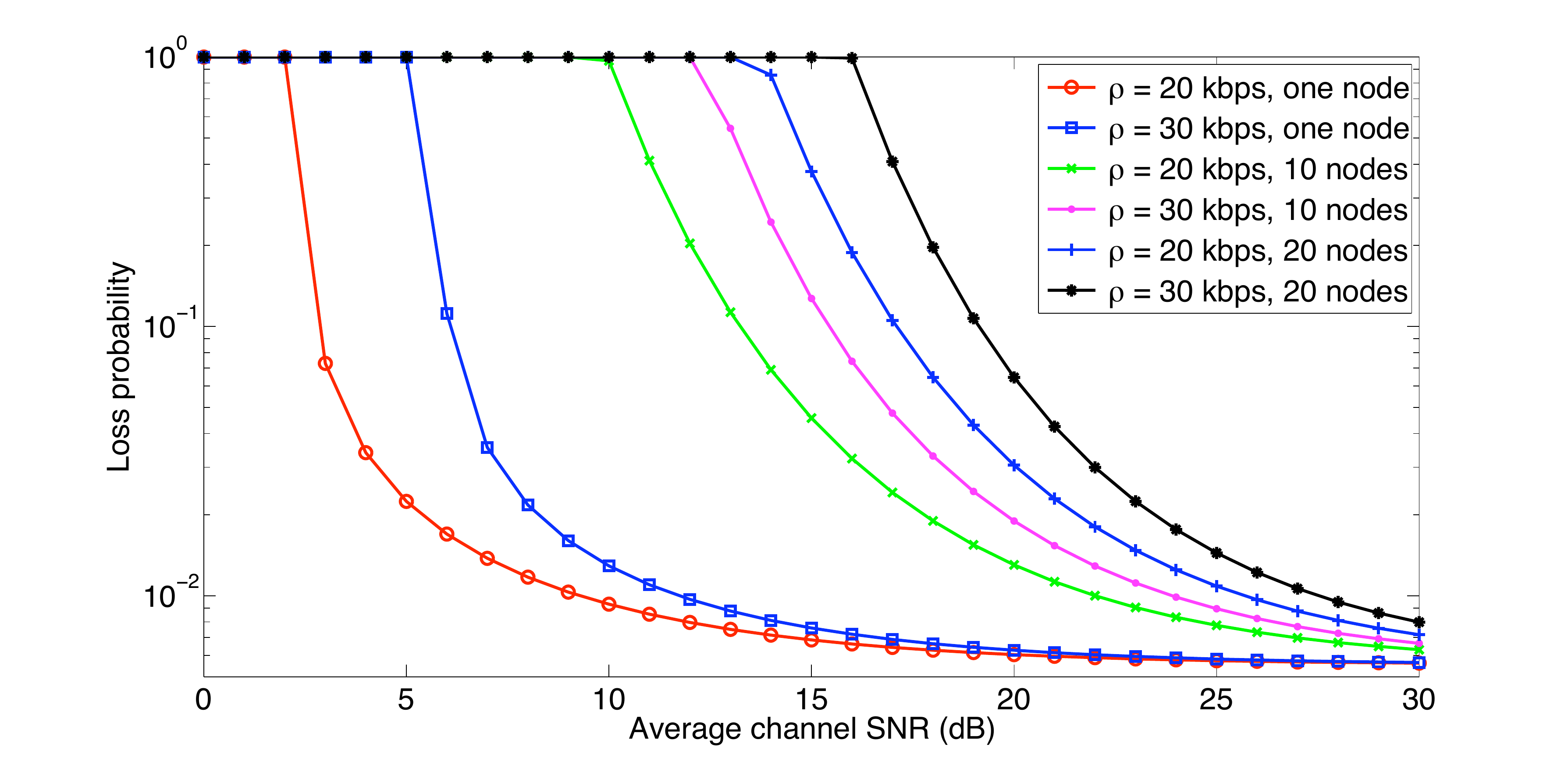}
\vspace{-2mm}
\caption{Loss probability ($\eps(b)$) vs. average channel  SNR ($\bar{\gamma}$) for  multi-hop  Rayleigh fading channels  for  $N=1,10$ and $20$, with buffer size $200$~kb, \sr~bounded traffic with $\sigma(s) = 50$~kb and  $\rho(s) = 20$ or $30$~kbps, and $W=20$~kHz.}
\label{fig:loss-vs-snr-multihop}

\centering
\includegraphics [width=5in]{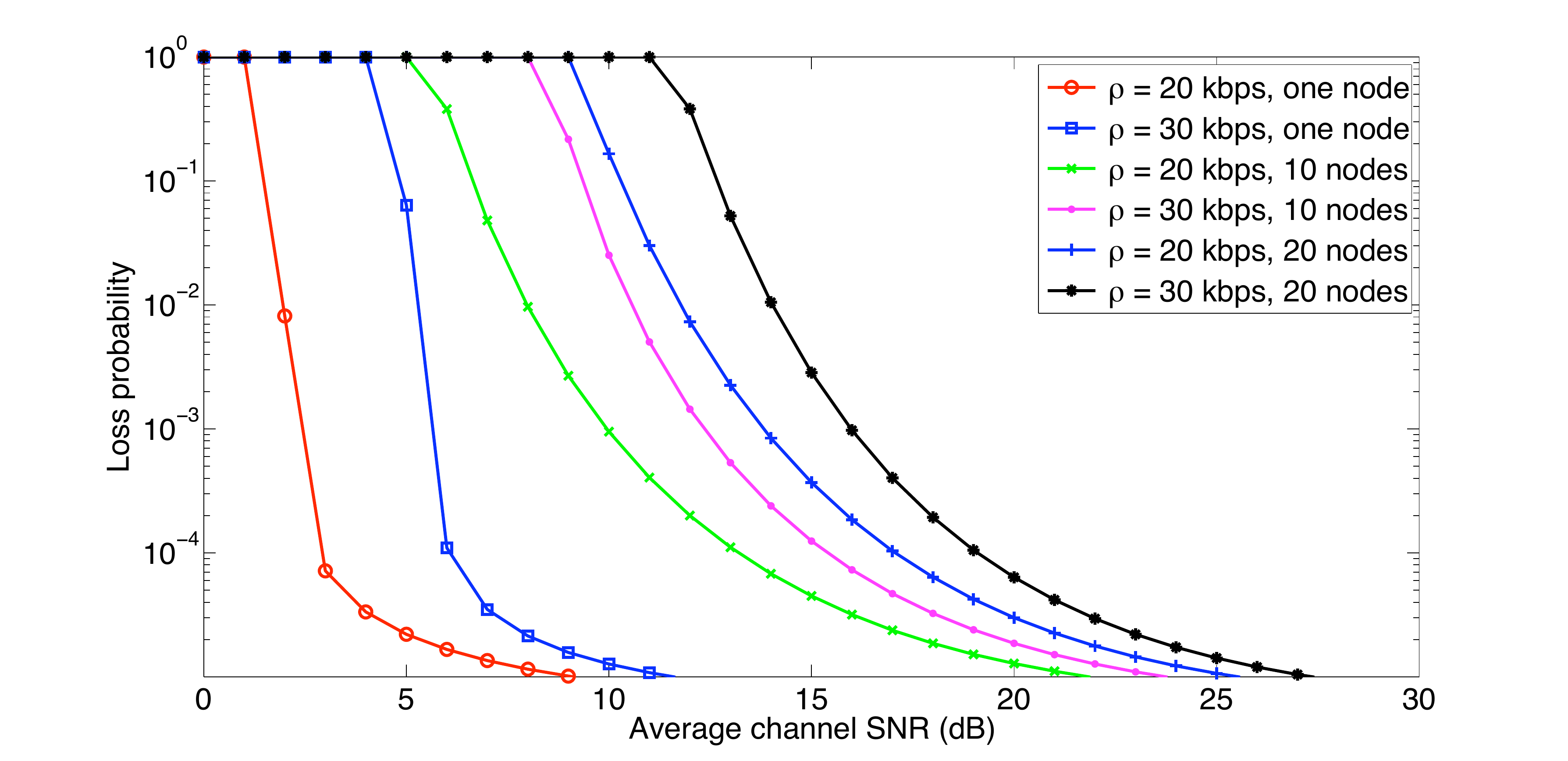}
\caption{Loss probability ($\eps(b)$) vs. average channel SNR ($\bar{\gamma} $) for   multi-hop   Rayleigh fading channels  for $N=1,10$ and $20$ with buffer size $ 400$~kb,   \sr~bounded traffic  with  $\sigma(s) = 50$~kb and  $\rho(s) = 20$ or $30$ kbps, and $W=20$~kHz.}
\label{fig:loss-vs-snr2-multihop}
\end{figure}

\begin{figure}
\centering
\includegraphics [width=5in]{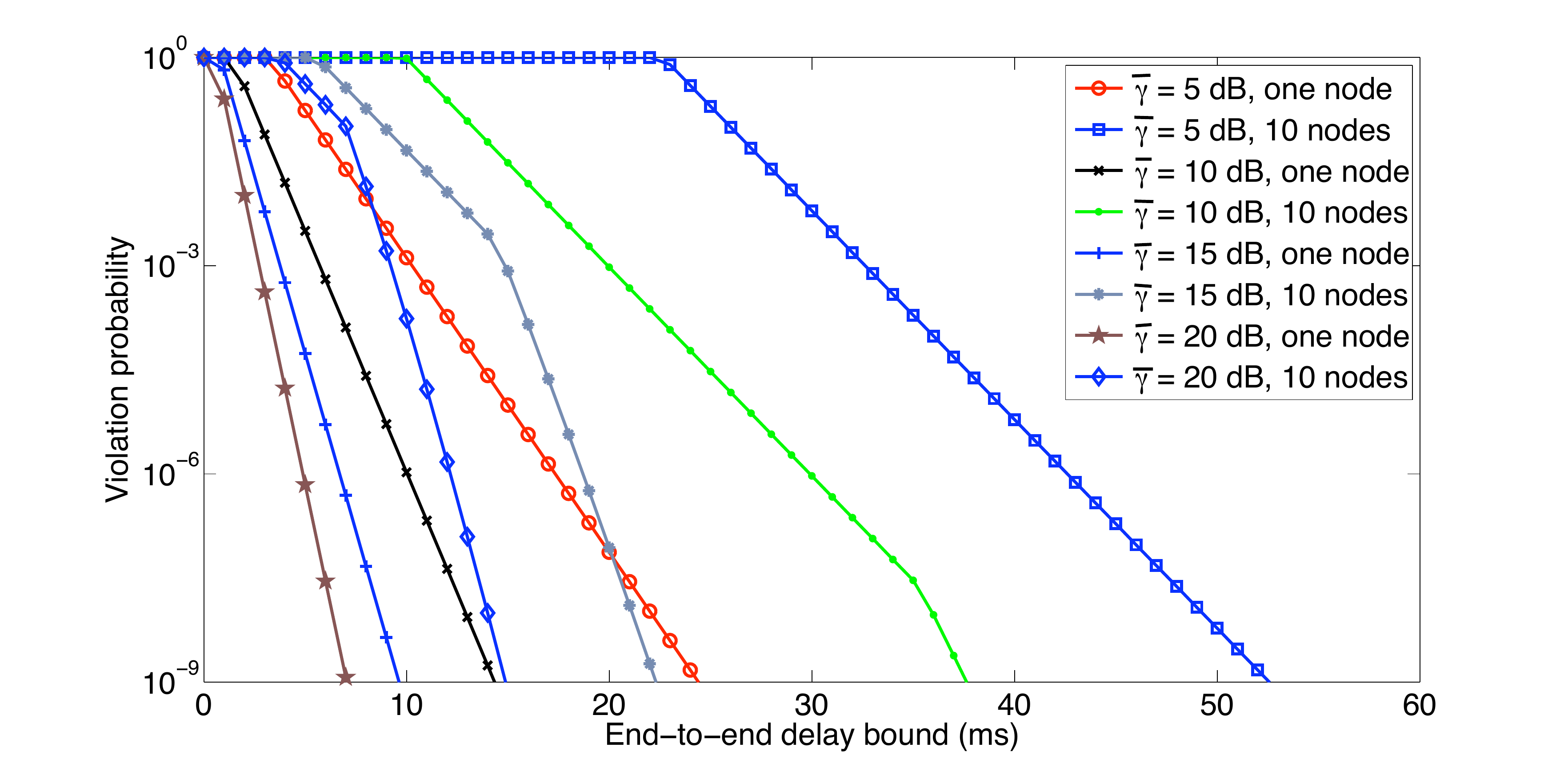}

\vspace{-2mm}
\caption{Delay bound violation probability ($\eps(w)$) vs. end-to-end delays for  multi-hop   Rayleigh fading channels   for   $N=1,10$,    $\bar{\gamma} = 5, 10, 15,20$~dB,    \sr ~bounded traffic  with $\sigma(s) = 50$~kb and $\rho(s) = 20$~kbps   and  $W=20$~kHz.}
\label{fig:viol-vs-delay-multihop}

\centering
\includegraphics [width=5in]{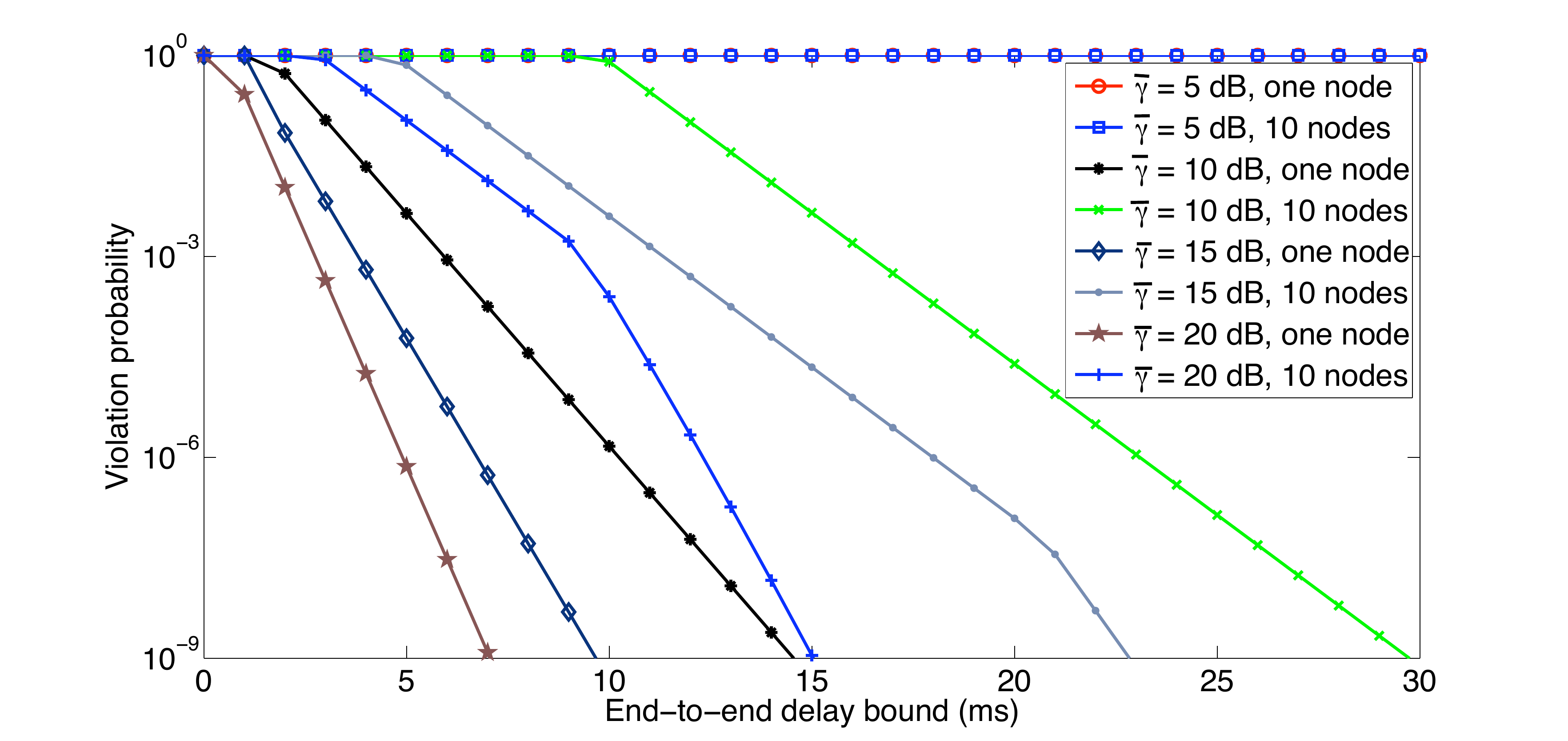}
\caption{Delay bound violation probability ($\eps(w)$) vs. end-to-end delay bound  for  multi-hop   Rayleigh fading channels   for   $N=1,10$,    $\bar{\gamma} = 5, 10, 15,20$~dB,    \sr ~bounded traffic  with $\sigma(s) = 50$~kb and $\rho(s) = 30$ kbps   and  $W=20$~kHz.}
\label{fig:viol-vs-delay-multihop2}
\end{figure}

\begin{figure}
\centering
\includegraphics [width=5in]{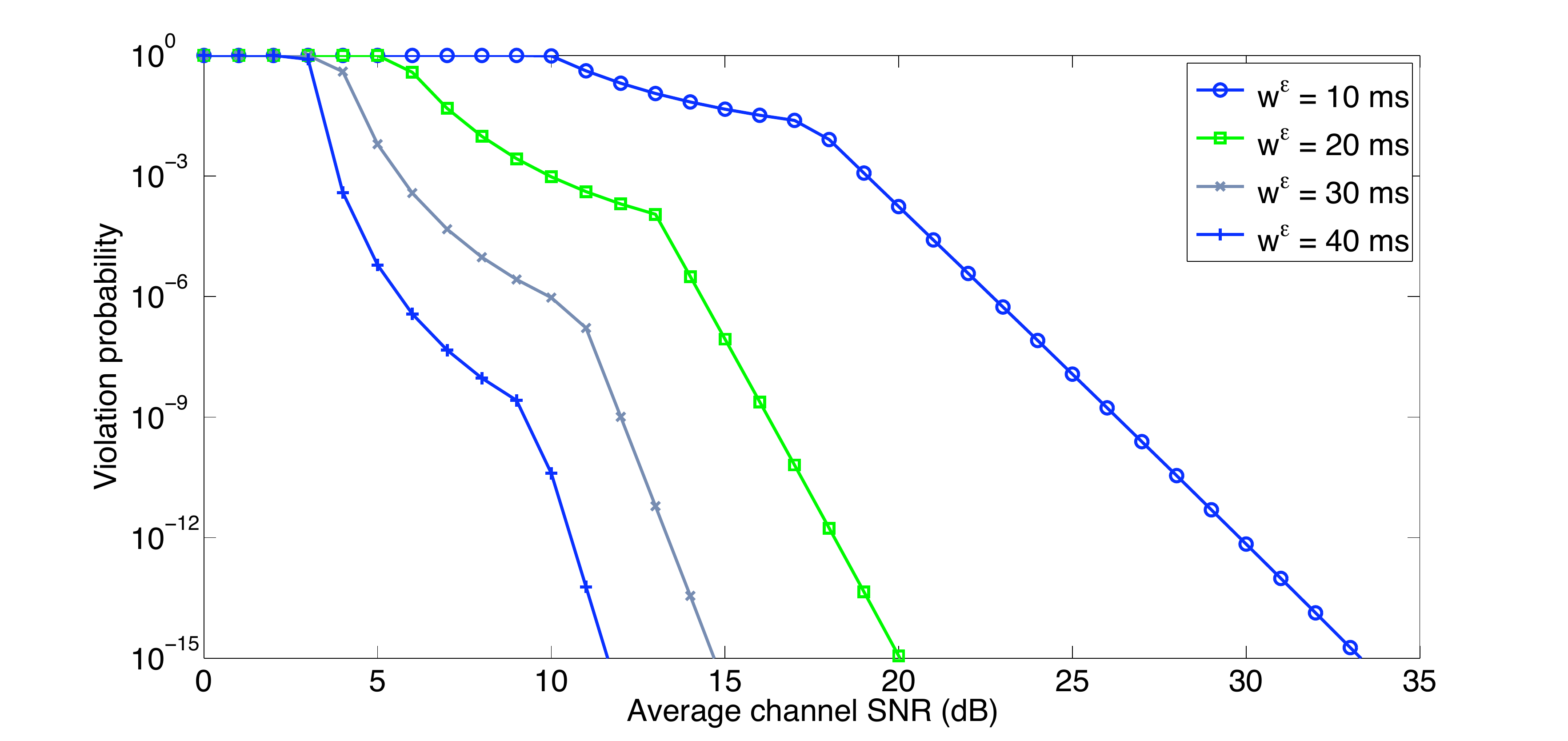}
\caption{Delay bound violation probability ($\eps(w)$) vs. average channel SNR ($\bar{\gamma} $)  for  multi-hop   Rayleigh fading channels   for   $N=10$,     $w^{\eps}=10,20,30,40$~ms,    \sr ~bounded traffic  with $\sigma(s) =~50$~kb and $\rho(s) = 20$~kbps   and  $W=20$~kHz.}
\label{fig:viol-vs-SNR-multihop}
\end{figure}

Suppose that buffer sizes are set to satisfy the end-to-end backlog. 
For a fixed buffer size $b_{\rm max}$, we can then 
use the probability $\P(B_{\rm net} (t) > b_{\rm max})$ as an estimate of the 
probability of dropped traffic, and refer to it as the {\it loss probability}. 
In Figs.~\ref{fig:loss-vs-snr-multihop} and~\ref{fig:loss-vs-snr2-multihop}, we depict the loss probability as 
a function of the average channel SNR, for values of $b_{\rm max}=200$~kb (Fig.~\ref{fig:loss-vs-snr-multihop}) and $b_{\rm max}=400$~kb (Fig.~\ref{fig:loss-vs-snr2-multihop}), 
for traffic with a rate of $\rho(s)= 20$~and $30$~kbps, and for $N=1,10$, and~$20$ nodes. 
The figure shows the minimum SNR needed to support a given loss probability 
is very sensitive to the number of network nodes. 
%

In Figs.~\ref{fig:viol-vs-delay-multihop} and~\ref{fig:viol-vs-delay-multihop2} we show the violation probability for 
given end-to-end delay bounds, where we compare the delays 
at a single node ($N=1$) with a multi-hop network ($N=10$) for 
different average channel SNR values,  
where we use Eqs.~\eqref{eq:Wbound} and~\eqref{eq:delay-M}. The traffic parameters are $\sigma(s) = 50$~kb for the burst, and $\rho(s) = 20$~kbps (in Fig.~\ref{fig:viol-vs-delay-multihop})  $\rho(s) = 30$~kbps (in Fig.~\ref{fig:viol-vs-delay-multihop2}). 
The figures illustrates that at sufficiently high SNR values, low delays 
are achieved even when traffic traverses 10~links. When the SNR is decreased, 
we can observe how the delay performance deteriorates in the multi-hop scenario.  
Note that the curves for the violation probability are essentially piecewise linear with two segments. This is caused by the different 
exponential decay rates, which follows from Eq.~\eqref{eq:delay-M}. 
Depending on where in the equation the minimum occurs, 
we obtain a faster or slower decay.

Next we investigate how the delay violation probability for given end-to-end delay bounds is impacted by the SNR of the channel. We consider a 
network with $N=10$~nodes, with the same parameters as before. 
The traffic parameters are $\sigma(s) = 50$~kb and $\rho(s) = 20$~kbps. 
We consider end-to-end delay bounds of $w^{\eps}=10,20,30,40$~ms. 
Figs.~\ref{fig:viol-vs-SNR-multihop} presents the results. An interesting 
observation is that the SNR required to support a given violation probability for a short delay bound  $w^{\eps}=10$~ms, the SNR 
requirement increase fast for low violation probabilities. Delay 
bounds  $w^{\eps}=20$~ms and higher can be supported with 
low violation probabilities even when the average SNR is small. 
 
The graphs presented here can be used in the planning of a multi-hop wireless network where predefined QoS bounds are desired for a given transmission rate.  Since the average SNR depends largely on the path loss, which, in turn, is a function of the transmission radius, 
the graphs could help with determining the maximum distance 
between nodes to support a desired transmission rate and QoS.

\clearpage
\begin{figure}
\centering
\includegraphics [width=3in]{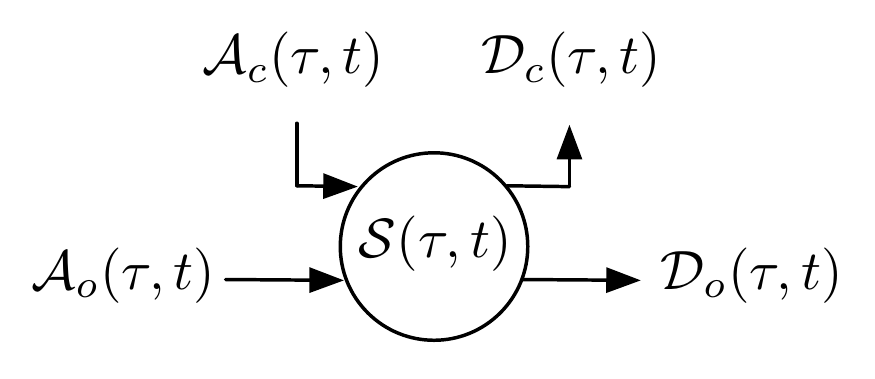}
\caption{Single Fading Channel with Cross Traffic.}
\label{fig:leftover-snr}
\end{figure}

\section{Fading Channels With Cross Traffic}

Consider a scenario in Fig.~\ref{fig:leftover-snr} where a through 
flow arriving to a fading channel shares the available bandwidth 
with  other flows.  We will refer to the traffic 
from these other flows as  {\it  cross traffic}. 
We use $\A_o(\tau,t)$ and $\A_c(\tau,t)$ to denote the 
SNR arrival processes of the through flow and the cross traffic, respectively, and
let $\D_o(\tau,t)$ and $\D_c(\tau,t)$ denote the corresponding
departure processes.  
In the SNR domain, cross traffic can  be viewed as 
reducing the channel 
capacity of the through flow  by generating interference. 

The following lemma states that, in the SNR domain, 
the service available to a through flow that experiences cross traffic 
at a channel can be expressed by a dynamic SNR server.   
\begin{lemma}  \label{lem:leftover}
Consider a channel with a through flow and cross traffic as shown in 
Fig.~\ref{fig:leftover-snr}. Assume that the channel provides a dynamic SNR server
to the aggregate of through flow and cross traffic, 
with service process $\S(\tau,t)$, i.e,
$$
\D_o(0,t) \cdot\D_c(0,t)\ge (\A_o\cdot \A_c)\conv \S (0,t)\,.
$$
Then
\begin{align*}
\S_o(\tau,t)
= \frac{\S(\tau, t)}{\A_c(\tau,t)}
\end{align*}
is a dynamic SNR server satisfying for all $t \geq 0$ that 
\begin{align}\label{eq:leftover-simple}
\D_o(0,t)\ge \A_o\conv \S_o(0,t)\,.
\end{align}
If, moreover, the service to the cross
flows satisfies the upper bound
$\D_c(0,t)\le \A_c\conv\S(0,t)$, then
\begin{align}\label{eq:leftover-plus}
\D_o(0,t)\ge \A_o\conv \max\{1,\S_o\}(0,t)\,.
\end{align}
\end{lemma}

\begin{proof} 
For any a  sample path, and any $t\ge 0$, we have 
\[
\D_o (0,t) \cdot \D_c(0,t) \ge
 \inf_{0\le \tau \le t} \{ (\A_o(0,\tau) \cdot \A_c(0,\tau))
\cdot \S(\tau, t)\} \ .
\]
Let $\tau^*$ be the point where the infimum is assumed.
Dividing by $\D_c(\tau, t)$, we obtain 
\begin{align*}
 \D_o(0,t) & \ge \frac{(\A_0 \cdot \A_c)\conv \S(0,t)}{
\D_c(0,t)}\\
& \ge \frac{\A_o(0,\tau^*) \cdot \A_c(0,\tau^*) 
\cdot \S(\tau^*, t)}{\D_c(0,t)}  \Big \} \\
 & \ge \Big \{ \A_o(0,\tau^*) \cdot\frac{  
\S(\tau^*, t)}{\A_c(\tau^*,t)}  \Big \}\,,
\end{align*}
where we used that $\D_c(0,t)\le \A_c(0,t)$ by causality.
This the first claim in Eq.~\eqref{eq:leftover-simple}.
To see the second claim, assume 
that $\D_c(0,t)\le \A_c\conv\S(0,t)$.  Then
$\D_c(0,t)\le \A_c(0,\tau^*)\cdot\S(\tau^*,t)$,
and therefore
\begin{align*}
 \D_o(0,t) 
& \ge \frac{\A_o(0,\tau^*) 
\cdot \A_c(0,\tau^*) \cdot \S(\tau^*, t)}{
A_c(0,\tau^*)\cdot\S(\tau^*,t)}  \\
&=\A_0(\tau^*,t)\,.
\end{align*}
Combining this with the first part of the proof, we obtain
\begin{align*}
\D_0(0,t)\ge \A_0(0,\tau^*)\cdot \max\{1,\S_o(\tau^*,t)\}\,,
\end{align*}
proving the second claim.
\end{proof}
Note that $\S_o(\tau,t)$ need not be monotone
in $t$, i.e., it may not lie in  $\F^+$.

\subsection{Performance Bounds With Cross Traffic}

We next estimate the service process available to the through flow 
across a cascade of channels.
Assuming that $\A_c$ and $\S$ are independent,
we obtain for the Mellin transform of
the  service process at a single node
\begin{align}\label{eq:leftover4} 
\M_{\S_o}(s,\tau,t) &=    \M_{ {\S }/{\A_c } } (s,\tau,t) \nonumber \\
&=    \M_{\S }  (s,\tau,t) \cdot  \M_{\A_c }  (2-s,\tau,t) \,.
\end{align}
These  service descriptions may be convolved,
using Lemmas~\ref{thm:conc} and~\ref{lem:MTConvDecon}
to obtain a bound for the Mellin transform of the service
provided by a cascade of fading channels to a flow
that experiences cross traffic at each node.

\begin{corollary}\label{cor:cascade-leftover}
Consider a network as in Fig.~\ref{fig:tandem}, where 
a through flow experiences cross traffic at each fading channel. 
Let the SNR service process at each channel 
be given by
Eq.~\eqref{eq:cumulative_ServiceCapacity}. Assume that
the arrival process of the cross traffic 
satisfies Eq.~\eqref{eq:srMA} 
with parameters $(\sigma_c(s),\rho_c(s))$.
Assume that arrivals from through flow and cross traffic, 
as well as the service processes at each channel 
are independent. Then, the service provided by the network
to the through flow satisfies $\D_o(0,t)\ge \A_o\conv\S_{o,\rm net}(0,t)$,
where the Mellin transform of $\S_{o,\rm net}(\tau,t)$
satisfies, for $s<1$,
\begin{align*}
\M_{\S_{o, \rm net}}(s,\tau,t)
\le e^{(1-s)\cdot N\sigma_c(1-s)}
\binom{N-1+t-\tau}{t-\tau} 
\cdot \bigl(\M_{g(\gamma)}(s)e^{(1-s)\cdot\rho_c(1-s)}\bigr)^{t-\tau}\,.
\end{align*}
\end{corollary}

\begin{proof}
For a single node ($N=1$), we estimate for $0\le \tau\le t$
and $s<1$
\begin{align*}
\M_{\S_o}(s,\tau,t) &= \M_\S(s,\tau,t)\cdot \M_{\A_c}(2-s,\tau,t)\\
&\le e^{(1-s)\cdot\sigma_c(1-s)}
\cdot\bigl(\M_{g(\gamma)}(s)\cdot e^{(1-s)\cdot\rho_c(1-s)}\bigr)^{t-\tau}\,,
\end{align*}
where we have used that $2-s>1$ for $s<1$.
By Lemma~\ref{thm:conc}, the service of the network is 
given by the \mx~convolution,
$\S_{o,\rm net}(\tau,t)=\S_{o,1}\conv \dots \conv \S_{o,N}(\tau,t)$.
We use Lemma~\ref{lem:MTConvDecon} to bound 
its Mellin transform by
\begin{align*}
\M_{\S_{\rm net}}(s,\tau,t) 
&\le \sum_{u_1\dots u_{N-1}} \prod_{n=1}^N M_{\S_{o,n}}(s,u_{n-1},u_n)\,,
\end{align*}
where the  sum runs over all sequences $u_0\le u_1\le\dots \le u_N$
with  $u_0=\tau$ and $u_N=t$. Collecting terms, we compute,
as in the proof of Corollary~\ref{cor:cascade},
that each product evaluates to the same term,
\begin{align*}
\prod_{n=1}^N M_{\S_{o,n}}(s,u_{n-1},u_n) 
&= e^{(1-s)\cdot N\sigma_c(1-s)}
\bigl(\M_{g(\gamma)}(s)\cdot e^{(1-s)\cdot(\rho_c(1-s)}\bigr)^{t-\tau}\,.
\end{align*}
Since the number of summands is $\binom{N-1+t-\tau}{t-\tau}$,
this proves the claim.
\end{proof}

We now consider Rayleigh fading channels, and assume that 
both through flow and cross traffic are \sr~bounded traffic 
with parameters $\sigma_o(s)$ and $\rho_o(s)$ for the through 
flow, and $\sigma_c(s)$ and $\rho_c(s)$ for the cross traffic.  
We next give end-to-end performance bounds,
using Theorem~\ref{thm:BLbound}. For the function
$\Msum(s,\tau, t)$, we compute as in Sec.~\ref{sec:rayleigh}.B
for $0\le \tau \le t$,
\begin{align*}
\Msum_{o,\rm net} (s, \tau, t) 
&\leq \ \frac{e^{s\cdot(\rho_o(s)(t-\tau)+\sigma_o(s)+N\sigma_c(s))}}{(1-V_o(s) )^N}
\,,
\end{align*}
where 
\begin{align*}
V_o(s)= e^{s\cdot(\rho_o(s)+\rho_c(s)} e^{1/\bar\gamma} 
\bar{\gamma}^{-s} \Gamma(1-s,\bar\gamma^{-1})\,.
\end{align*}
This computation is valid under the stability condition that
$V_o(s)<1$. Thus, we obtain that the output burstiness at 
the network egress gives 
\begin{eqnarray}\label{eq:LOMultiNodeOutput}
d^\eps_{o,\rm net}(\tau,t)  
 \le 
 \inf_{s>0} \Bigg \{ \sigma_o(s)+N\sigma_c(s) 
+\rho_o(s)(t-\tau)
 -\frac{1}{s} \Bigg [    N \log  \big ( 1 -      V_o(s) \big ) +      \log \eps     \Bigg ]   \Bigg \} \,.  \nonumber 
 \end{eqnarray}  
For the end-to-end backlog of the through flow, we obtain 
\begin{eqnarray}\label{eq:LOMultiNodeBacklog}
b^\eps_{o,\rm net}(t)
\le \inf_{s>0} \Bigg \{ \sigma_o(s)+N\sigma_c(s) 
-\frac{1}{s} \Bigg [    N \log  \big ( 1 - 
    V_o(s) \big ) +     \log \eps     \Bigg ]    \Bigg \} \,.  \nonumber 
 \end{eqnarray}  
Similarly, for the delay bound, we estimate for $w\ge 0$
\begin{align*}
\notag
\Msum_{\rm net} (s, t+w, t) 
&\leq \ e^{s(-\rho_o(s) w +\sigma_o(s)+N\sigma_c(s))}
\sum_{u=w}^\infty \binom{N-1+u}{u}(V_o(s))^u\\
&\leq \ 
\inf_{s>0}\left\{
\frac{e^{s(-\rho_o(s) w +\sigma_o(s)+N\sigma_c(s))}
}{(1-V_o(s) )^N} 
\cdot\min \left\{1,(V_o(s))^w w^{N-1} \right\}\right\} \,.
\end{align*}


\begin{figure}[h]
\centering
\includegraphics [width=5in]{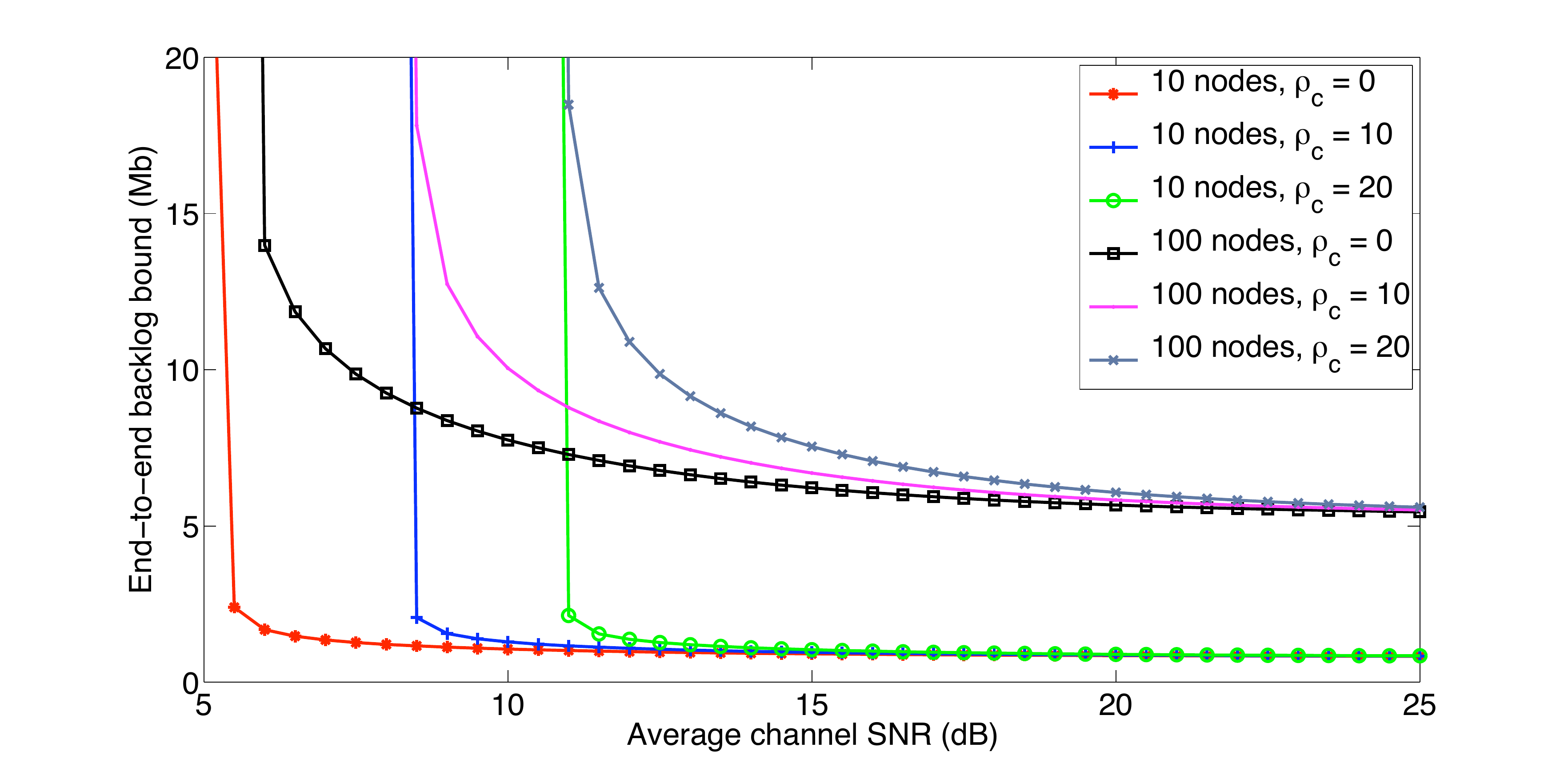}
\caption{End-to-end backlog bound ($b_{o,{\rm net}}^{\eps}$) vs. average channel SNR ($\bar{\gamma} $) for   multi-hop    Rayleigh fading channels  with   $\eps =10^{-4}$, $N=10$~and~$100$, \sr~bounded traffic with $\sigma_o(s) = 50$~kb,  $\rho_o(s) = 30$~kbps, \sr ~bounded cross traffic with $\sigma_c(s) = 50$~kb and $\rho_c(s)=0, 10, 20$, and $W=20$~kHz.}
\label{fig:BacklogVsSNR_wCT}

\vspace{-4mm}
%
\centering
\includegraphics [width=5in]{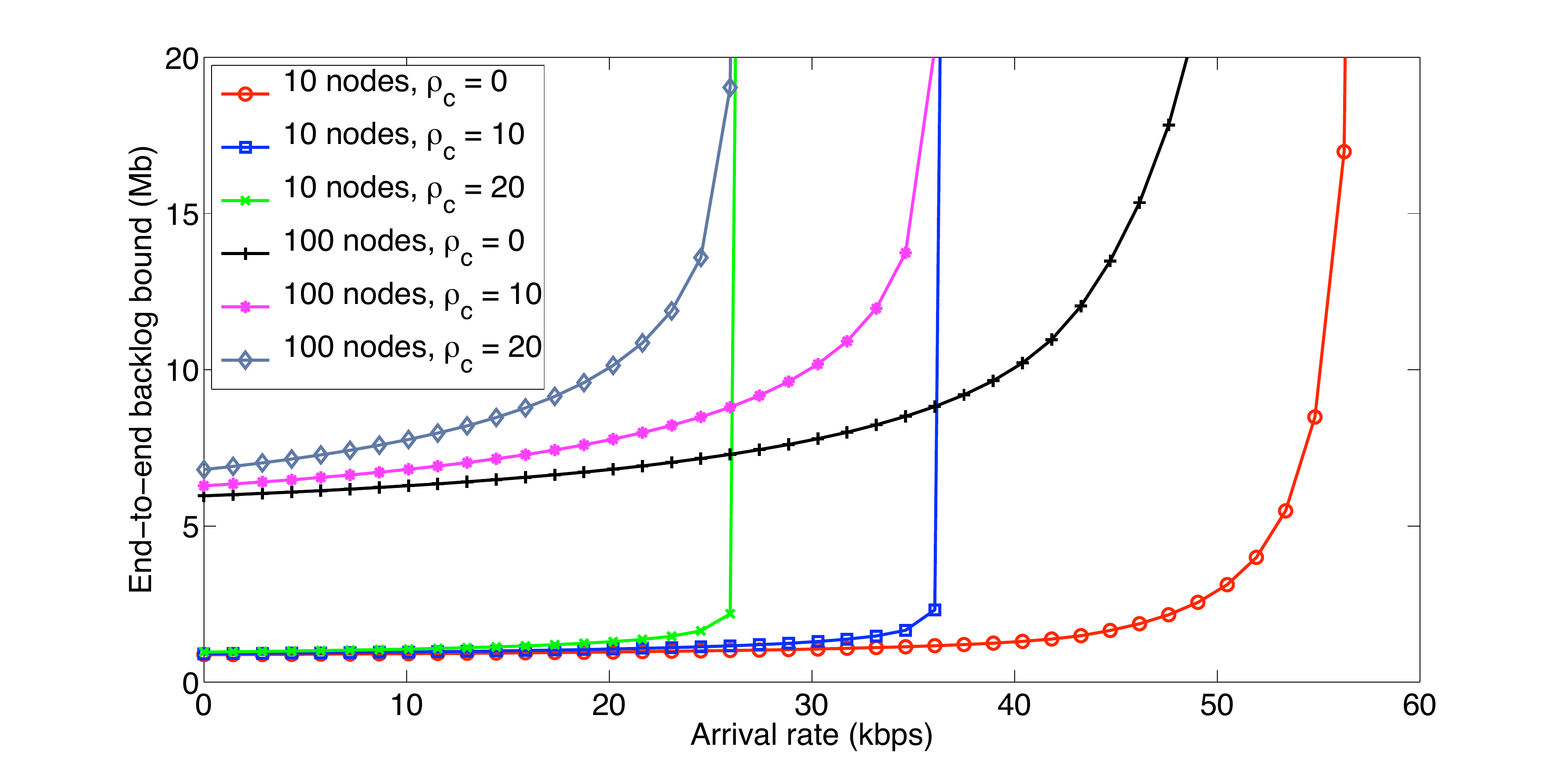}
\caption{End-to-end backlog bound  ($b_{o,{\rm net}}^{\eps}$) vs. arrival rate ($\rho_o(s)$) for   multi-hop   Rayleigh fading channels with   $\eps =10^{-4}$, \sr ~bounded traffic with $\sigma_o(s) = 50$~kb, \sr ~bounded cross traffic with $\sigma_c(s) = 50$~kb and $\rho_c(s)=0, 10, 20$, $\bar{\gamma} = 10$~dB, and $W=20$~kHz.}
\label{fig:BacklogVsRho_wCT1}

\vspace{-4mm}
%
\centering
\includegraphics [width=5in]{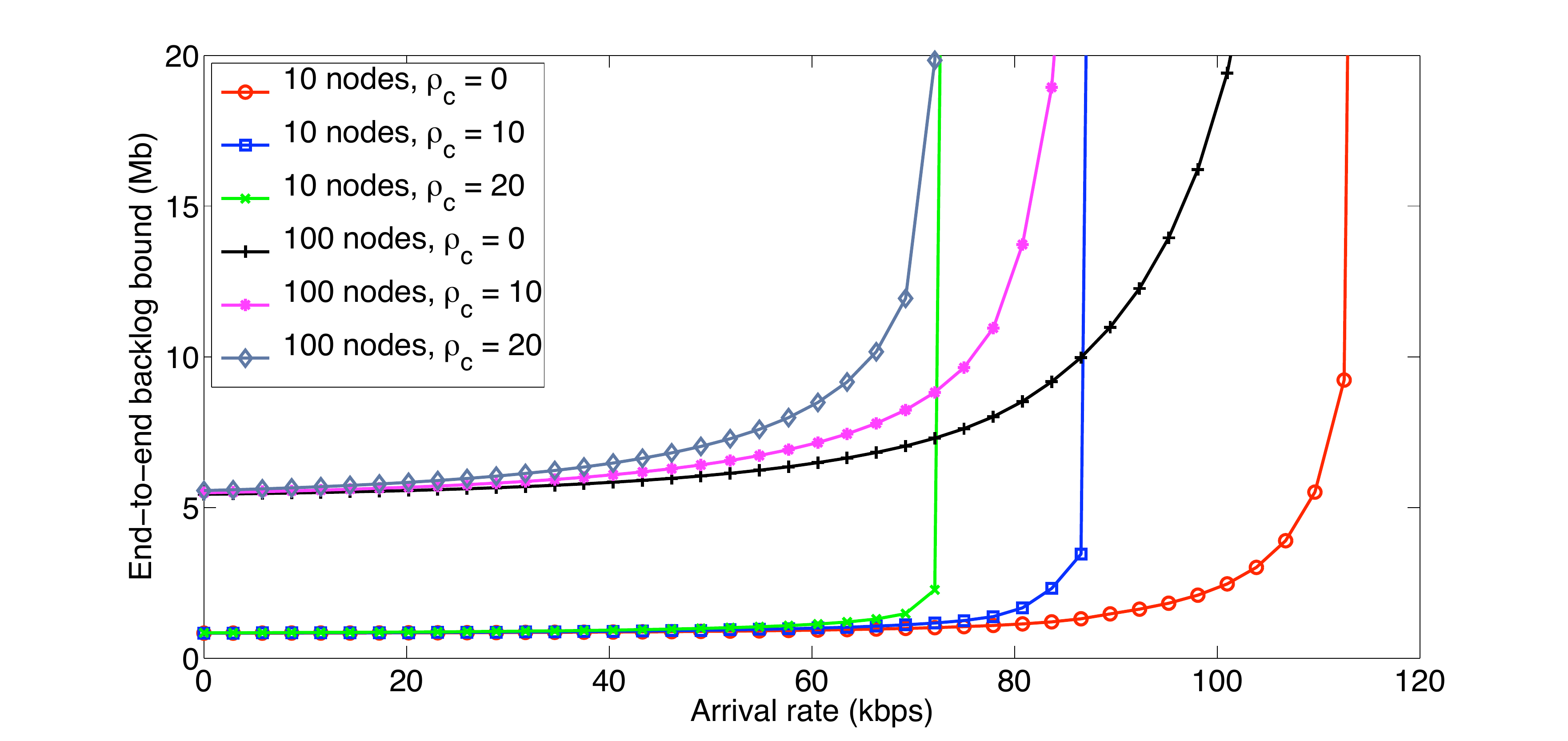}
\caption{End-to-end backlog bound  ($b_{o,{\rm net}}^{\eps}$) vs. arrival rate ($\rho_o(s)$) for   multi-hop   Rayleigh fading channels with   $\eps =10^{-4}$, \sr ~bounded traffic with $\sigma_o(s) = 50$~kb, \sr ~bounded cross traffic with $\sigma_c(s) = 50$~kb and $\rho_c(s)=0, 10, 20$, $\bar{\gamma} = 20$~dB, and $W=20$~kHz.}
\label{fig:BacklogVsRho_wCT2}
\end{figure}

\subsection{Numerical Examples}
\label{sec:numerical-cross}

We now present numerical examples for Rayleigh fading channels. 
We use the same traffic and channel model as in the numerical examples 
in Sec.~\ref{sec:numerical}. The cross traffic is  \sr~bounded traffic 
with parameters $\sigma_c(s)$ and  $\rho_c(s)$. As before, we assume
that cross traffic is deterministically bounded by using 
fixed values for $\sigma_c(s)$ and $\rho_c(s)$. Consequently, 
there is no statistical multiplexing between through and cross traffic. 
Throughout, the  violation 
probability is  set to $\eps =10^{-4}$.

In Fig.~\ref{fig:BacklogVsSNR_wCT} we show the end-to-end backlog bound $b_{o,{\rm net}}^{\eps}$ as a function of the average channel SNR 
$\bar{\gamma}$. The through flow has fixed parameters 
$\sigma_o(s) = 50$~kb and  $\rho_o(s) = 30$~kbps, and the 
cross traffic has parameters $\sigma_c(s) = 50$~kb and $\rho_c(s)=0, 10, 20$. We consider networks with $N=10$ and $N=100$ nodes. 
The graph illustrates how the bursts of the cross traffic contribute 
to the backlog bound. There is an additive component for each 
traversed channel, which explains the difference for the backlog 
with $10$ and $100$~channels.  The minimum required SNR needed for 
stability appears  less 
sensitive to the number of channels.

In Figs.~\ref{fig:BacklogVsRho_wCT1} and~\ref{fig:BacklogVsRho_wCT2} we again evaluate the end-to-end backlog bound $b_{o, {\rm net}}^{\eps}$. Here, we keep the channel SNR values 
constant at $\bar{\gamma}=10$ (Fig.~\ref{fig:BacklogVsRho_wCT1}) and 
$\bar{\gamma}=20$ (Fig.~\ref{fig:BacklogVsRho_wCT2}). We vary the rate of the through flow 
$\rho_o(s)$ on the x axis, and plot graphs for different values of 
$\rho_c(s)=0, 10, 20$. We again consider $N=10$ and $N=100$.
The outcomes are as seen in the previous figure. We observe the 
effects of the additive component contributed by each node. 
Also, we observe, at least for lower transmission rates, 
that the stability  of end-to-end backlogs is not sensitive to the number of traversed channels.

\clearpage

\section{Conclusion}
\label{sec:conclusions}
We have developed a novel network calculus 
that can incorporate  
fading channel distributions, without the need for 
secondary models, such as FSMC. We use the calculus to compute 
statistical bounds on 
delay and backlog of multi-hop wireless networks with 
fading channels.  
We took a fresh point of view, where the descriptions of 
the arrivals  and the 
fading channels reside in different domains, referred to as 
bit domain and SNR domain. 
We found that by mapping arrival processes to the SNR domain, 
an end-to-end analysis with fading channels becomes tractable. 
An important discovery was that arrivals and 
service in the SNR domain obey the laws of a \mx~dioid algebra.  

The analytical framework developed in this paper appears 
suitable to study a broad class of fading channels  
and their impact on the network-layer performance in wireless networks. 
Even though we computed numerical examples for very simplified 
networks, in particular, we made strong independence assumptions 
for the fading channels, 
our \mx~network calculus is applicable to networks where these 
assumptions are relaxed. Generalizing our framework 
and obtaining a more profound understanding of the dioid algebra 
and computational methods in the SNR domain is the subject of future research.


\begin{thebibliography}{10}

\bibitem{Ngatched}
T.~N.~A. Alfa and J.~Cai.
\newblock A weighted queue-based model for correlated {R}ayleigh and {R}ician
  fading channels.
\newblock {\em IEEE Trans. on Communications}, 59(11):3049--3058, Nov. 2011.

\bibitem{Amarasuriya11}
G.~Amarasuriya, C.~Tellambura, and M.~Ardakani.
\newblock Asymptotically-exact performance bounds of {AF} multi-hop relaying
  over {N}akagami fading.
\newblock {\em IEEE Trans. Commun.}, 59(4):962--967, April 2011.

\bibitem{Bisnik}
N.~Bisnik and A.~A. Abouzeid.
\newblock Queuing network models for delay analysis of multihop wireless ad hoc
  networks.
\newblock {\em Elsevier Ad Hoc Networks}, 7(1):79--97, January 2009.

\bibitem{Burchard_ToIT06}
A.~Burchard, J.~Liebeherr, and S.~Patek.
\newblock A min-plus calculus for end-to-end statistical service guarantees.
\newblock {\em IEEE Trans. on Information Theory}, 52(9):4105--4114, September
  2006.

\bibitem{CSChang}
C.-S. Chang.
\newblock {\em Performance guarantees in communication networks}.
\newblock Springer Verlag, 2000.

\bibitem{ChangCruz99}
C.-S. Chang and R.~L. Cruz.
\newblock A time varying filtering theory for constrained traffic regulation
  and dynamic service guarantees.
\newblock In {\em Proc. IEEE Infocom}, pages 63--70, 1999.

\bibitem{Chen_Yang_Darwazeh}
Y.~Chen, Y.~Yang, and I.~Darwazeh.
\newblock A cross-layer analytical model of end-to-end delay performance for
  wireless multi-hop environments.
\newblock In {\em Proc. IEEE Globecom}, pages 1--6, December 2010.

\bibitem{Choe-Shroff-MVA}
J.~Choe and N.~B. Shroff.
\newblock A central-limit-theorem-based approach for analyzing queue behavior
  in high-speed networks.
\newblock {\em IEEE/ACM Transactions on Networking}, 6:659--671, 1998.

\bibitem{SigmetricsCiucu11}
F.~Ciucu.
\newblock Non-asymptotic capacity and delay analysis of mobile wireless
  networks.
\newblock In {\em Proc. ACM Sigmetrics}, pages 359--360, June 2011.

\bibitem{CiucuAlerton10}
F.~Ciucu, P.~Hui, and O.~Hohlfeld.
\newblock Non-asymptotic throughput and delay distributions in multi-hop
  wireless networks.
\newblock In {\em Proc. 48th Annual Allerton Conf. on Communication, Control,
  and Computing}, pages 662--669, Sep. 2010.

\bibitem{Davies}
B.~Davies.
\newblock {\em Integral transforms and their applications}.
\newblock Springer-Verlag, NY, 1978.

\bibitem{Elliott}
E.~Elliott.
\newblock Estimates of error rates for codes on burst-noise channels.
\newblock {\em Bell System Technical Journal}, 42(5):1977--1997, September
  1963.

\bibitem{FidlerMGF}
M.~Fidler.
\newblock An end-to-end probabilistic network calculus with moment generating
  functions.
\newblock In {\em Proc. IEEE IWQoS}, pages 261--270, June 2006.

\bibitem{Fidler-Fading}
M.~Fidler.
\newblock A network calculus approach to probabilistic quality of service
  analysis of fading channels.
\newblock In {\em Proc. IEEE Globecom}, pages 1--6, Nov. 2006.

\bibitem{Giacomazzi-Saddemi}
P.~Giacomazzi and G.~Saddemi.
\newblock Bounded-variance network calculus: Computation of tight
  approximations of end-to-end delay.
\newblock In {\em Proc. IEEE ICC}, pages 170--175, May 2008.

\bibitem{Gilbert}
E.~Gilbert.
\newblock Capacity of a burst-noise channel.
\newblock {\em Bell System Technical Journal}, 39(9):1253--1265, September
  1960.

\bibitem{Hasna_Alouini_2003}
M.~O. Hasna and M.~S. Alouini.
\newblock Outage probability of multihop transmission over {N}akagami fading
  channels.
\newblock {\em IEEE Commun. Lett.}, 7(5):216--218, 2003.

\bibitem{krunz2004}
M.~Hassan, M.~M. Krunz, and I.~Matta.
\newblock Markov-based channel characterization for tractable performance
  analysis in wireless packet networks.
\newblock {\em IEEE Trans. Wireless Commun.}, 3(3):821--831, 2004.

\bibitem{ishizaki2007}
F.~Ishizaki and H.~G. Uk.
\newblock Queuing delay analysis for packet schedulers with/without multiuser
  diversity over a fading channel.
\newblock {\em IEEE Trans. Veh. Technol.}, 56(5):3220--3227, 2007.

\bibitem{Jiang:2005}
Y.~Jiang and P.~J. Emstad.
\newblock Analysis of stochastic service guarantees in communication networks:
  a server model.
\newblock In {\em Proc. IWQoS, Springer Lecture Notes in Computer Science
  3552}, pages 233--245, June 2005.

\bibitem{Jiang-Book}
Y.~Jiang and Y.~Liu.
\newblock {\em Stochastic network calculus}.
\newblock Springer, 2008.

\bibitem{Kelly}
F.~Kelly.
\newblock Notes on effective bandwidths.
\newblock In {\em Stochastic Networks: Theory and Applications. (Editors: F.P.
  Kelly, S. Zachary and I.B. Ziedins) Royal Statistical Society Lecture Notes
  Series, 4}, pages 141--168. Oxford University Press, 1996.

\bibitem{Le_Ekram08}
L.~Le and E.~Hossain.
\newblock Tandem queue models with applications to {QoS} routing in multihop
  wireless networks.
\newblock {\em IEEE Trans. Mobile Comput.}, 7:1025--1040, 2008.

\bibitem{Le_Nguyen_HossainWCNC07}
L.~Le, A.~Nguyen, and E.~Hossain.
\newblock A tandem queue model for performance analysis in multihop wireless
  networks.
\newblock In {\em Proc. IEEE WCNC}, pages 2981--2985, March 2007.

\bibitem{leBoudec}
J.~Le~Boudec and P.~Thiran.
\newblock {\em Network calculus: a theory of deterministic queuing systems for
  the Internet}.
\newblock Lecture Notes in Computer Science 2050. Springer, 2001.

\bibitem{LiBuLi07}
C.~Li, A.~Burchard, and J.~Liebeherr.
\newblock A network calculus with effective bandwidth.
\newblock {\em IEEE/ACM Trans. on Networking}, 15(6):1442--1453, Dec. 2007.

\bibitem{Chengzhi-Che-Li}
C.~Li, H.~Che, and S.~Li.
\newblock A wireless channel capacity model for quality of service.
\newblock {\em IEEE Trans. Wireless Commun.}, 6(1):356--366, 2007.

\bibitem{Mahmood_Rizk_Jiang}
K.~Mahmood, A.~Rizk, and Y.~Jiang.
\newblock On the flow-level delay of a spatial multiplexing {MIMO} wireless
  channel.
\newblock In {\em Proc. IEEE ICC}, pages 1--6, June 2011.

\bibitem{MahmoodArxiv}
K.~Mahmood, M.~Vehkaper{\"a}, and Y.~Jiang.
\newblock Cross-layer modeling of randomly spread CDMA using stochastic network
  calculus.
\newblock {\em http://arxiv.org/abs/1105.0215}, 2011.

\bibitem{Proakis2007}
J.~Proakis and M.~Salehi.
\newblock {\em Digital Communications (5th Ed.)}.
\newblock McGraw Hill, 2007.

\bibitem{Sadeghi}
P.~Sadeghi, R.~Kennedy, P.~Rapajic, and R.~Shams.
\newblock Finite-state {M}arkov modeling of fading channels -- a survey of
  principles and applications.
\newblock {\em IEEE Signal Process. Mag.}, 25(5):57--80, 2008.

\bibitem{Tsiftsis08}
T.~A. Tsiftsis.
\newblock Performance of wireless multihop communications systems with
  cooperative diversity over fading channels.
\newblock {\em Wiley Int. J. Communication Systems}, 21(5):559--565, 2008.

\bibitem{Verticale-Q2SWinet}
G.~Verticale.
\newblock A closed-form expression for queuing delay in {R}ayleigh fading
  channels using stochastic network calculus.
\newblock In {\em Proc. ACM Q2SWinet '09}, pages 8--12, 2009.

\bibitem{Verticale:2009}
G.~Verticale and P.~Giacomazzi.
\newblock An analytical expression for service curves of fading channels.
\newblock In {\em Proc. IEEE Globecom}, pages 635--640, Nov. 2009.

\bibitem{WangMoayeri}
H.~Wang and N.~Moayeri.
\newblock Finite-state {M}arkov channel - a useful model for radio
  communication channels.
\newblock {\em IEEE Trans. Veh. Technol.}, 44(1):163--171, 1995.

\bibitem{Wang_CorrNakagami}
Q.~Wang, D.~Wu, and P.~Fan.
\newblock Effective capacity of a correlated {N}akagami-m fading channel.
\newblock {\em Wirel. Commun. Mob. Comp.. DOI: 10.1002/wcm.1048}, 2011.

\bibitem{Wang_CorrRayleigh}
Q.~Wang, D.~Wu, and P.~Fan.
\newblock Effective capacity of a correlated {R}ayleigh fading channel.
\newblock {\em Wirel. Commun. Mob. Comp.. DOI: 10.1002/wcm.945},
  11(11):1485--1494, 2011.

\bibitem{Wu}
D.~Wu and R.~Negi.
\newblock Effective capacity: a wireless link model for support of quality of
  service.
\newblock {\em IEEE Trans. Wireless Commun.}, 2(4):630--643, 2003.

\bibitem{Wu-MNA06}
D.~Wu and R.~Negi.
\newblock Effective capacity-based quality of service measures for wireless
  networks.
\newblock {\em J. on Mobile Networking and Applications}, 11(1):91--99,
  February 2006.

\bibitem{ZhengIET2011}
K.~Zheng, L.~Lei, Y.~Wang, Y.~Lin, and W.~Wang.
\newblock Quality-of-service performance bounds in wireless multi-hop relaying
  networks.
\newblock {\em IET Communications}, 5(1):71--78, Jan. 2011.

\bibitem{Zhong_Alajaji_Takahara}
L.~Zhong, F.~Alajaji, and G.~Takahara.
\newblock A binary communication channel with memory based on a finite queue.
\newblock {\em IEEE Trans. Information Theory}, 53(8):2815--2840, August 2007.

\bibitem{Zorzi-ICPUC}
M.~Zorzi, R.~Rao, and L.~Milstein.
\newblock On the accuracy of a first-order {M}arkov model for data transmission
  on fading channels.
\newblock In {\em Proc. IEEE ICPUC}, pages 211--215, November 1995.

\end{thebibliography}

\end{document}